%% file: arxiv_merged.tex
\documentclass[journal]{IEEEtran}

\usepackage{amsmath,amssymb,amsfonts,amsthm}
\usepackage{mathtools}
\usepackage{graphicx}
\usepackage{textcomp}
\usepackage{xcolor}
\usepackage[noadjust]{cite}
\usepackage{booktabs}
\usepackage{algorithm}
\usepackage{algpseudocode}
\usepackage[bookmarks,colorlinks]{hyperref}

\usepackage{acronym}
\usepackage{enumitem}
\IfFileExists{myshorts.tex}{\input{myshorts}}{}
\IfFileExists{Preamble.tex}{\input{Preamble.tex}}{%
  \acrodef{gsp}[GSP]{graph signal processing}%
  \acrodef{gft}[GFT]{graph Fourier transform}%
  \acrodef{evd}[EVD]{eigenvalue decomposition}%
}

\providecommand{\define}{\triangleq}
\providecommand{\tr}{\operatorname{tr}}
\providecommand{\diag}{\operatorname{diag}}
\providecommand{\rank}{\operatorname{rank}}
\providecommand{\range}{\operatorname{range}}

\providecommand{\argmin}{\operatorname*{arg\,min}}
\providecommand{\odot}{\mathbin{\odot}}
\providecommand{\oslash}{\mathbin{\oslash}}
\providecommand{\onevec}{\mathbf 1}
\providecommand{\xvec}{\mathbf x}
\providecommand{\avec}{\mathbf a}
\providecommand{\yvec}{\mathbf y}
\providecommand{\wvec}{\mathbf w}

\providecommand{\hvec}{\mathbf h}

\providecommand{\Lmat}{\mathbf L}
\providecommand{\Vmat}{\mathbf V}
\providecommand{\Wmat}{\mathbf W}
\providecommand{\Lambdamat}{\boldsymbol\Lambda}
\providecommand{\Imat}{\mathbf I}

\providecommand{\Amat}{\mathbf A}
\providecommand{\Bmat}{\mathbf B}
\providecommand{\Cmat}{\mathbf C}
\providecommand{\Hmat}{\mathbf H}

\providecommand{\Gmat}{\mathbf G}

\providecommand{\Dmat}{\mathbf D}

\providecommand{\mubx}{\boldsymbol\mu_x}
\providecommand{\muby}{\boldsymbol\mu_y}
\providecommand{\Cxx}{\mathbf C_{xx}}
\providecommand{\Cxy}{\mathbf C_{xy}}
\providecommand{\Cyy}{\mathbf C_{yy}}
\providecommand{\alphavec}{\boldsymbol\alpha}
\providecommand{\etavec}{\boldsymbol\eta}
\providecommand{\omegavec}{\boldsymbol\omega}

\providecommand{\nubf}{\boldsymbol\nu}

\providecommand{\Ccal}{\mathcal C}

\providecommand{\Rcal}{\mathcal R}
\providecommand{\Ncal}{\mathcal N}

\newcommand{\muxm}[1]{\boldsymbol\mu_x^{(#1)}}
\newcommand{\muym}[1]{\boldsymbol\mu_y^{(#1)}}
\newcommand{\Cxxm}[1]{\mathbf C_{xx}^{(#1)}}
\newcommand{\Cxym}[1]{\mathbf C_{xy}^{(#1)}}
\newcommand{\Cyym}[1]{\mathbf C_{yy}^{(#1)}}
\newcommand{\dxym}[1]{\mathbf d_{xy}^{(#1)}}
\newcommand{\dyym}[1]{\mathbf d_{yy}^{(#1)}}
\newcommand{\fstar}[1]{\mathbf f_{#1}^{\star}}
\newcommand{\muxj}[1]{\boldsymbol\mu_{x,#1}}
\newcommand{\muyj}[1]{\boldsymbol\mu_{y,#1}}
\newcommand{\Cxxj}[1]{\mathbf C_{xx,#1}}
\newcommand{\Cxyj}[1]{\mathbf C_{xy,#1}}
\newcommand{\Cyyj}[1]{\mathbf C_{yy,#1}}
\newcommand{\Astar}{\mathbf A_{\rm L}^{\star}}

\newcommand{\SmatD}{\mathbf S_{\rm D}}
\newcommand{\rvecD}{\mathbf r_{\rm D}}

\newtheorem{theorem}{Theorem}
\newtheorem{proposition}{Proposition}
\newtheorem{corollary}{Corollary}

\theoremstyle{remark}
\newtheorem{remark}{Remark}

\title{Graph Signal Estimation With Unknown Topology}

\author{Tirza Routtenberg, \IEEEmembership{Senior Member, IEEE}%
\thanks{The author is with the School of Electrical and Computer Engineering and with the Goldman Sonnenfeldt School of Sustainability and Climate Change
Ben-Gurion University of the Negev, Beer-Sheva 84105, Israel, e-mail: tirzar@bgu.ac.il. Parts of this work were presented at the CISS Conference, 2025~\cite{Zeltzer_Routtenberg_CISS}. This research was supported by the ISRAEL SCIENCE FOUNDATION (Grant No. 1148/22).}
\vspace{-0.5cm}}

\providecommand{\tr}{\operatorname{tr}}
\providecommand{\diag}{\operatorname{diag}}
\providecommand{\argmin}{\operatorname*{arg\,min}}
\providecommand{\argmax}{\operatorname*{arg\,max}}
\providecommand{\xvec}{\mathbf x}
\providecommand{\yvec}{\mathbf y}
\providecommand{\hvec}{\mathbf h}
\providecommand{\bvec}{\mathbf b}

\providecommand{\Lmat}{\mathbf L}
\providecommand{\Vmat}{\mathbf V}
\providecommand{\Imat}{\mathbf I}
\providecommand{\Amat}{\mathbf A}
\providecommand{\Hmat}{\mathbf H}
\providecommand{\Kmat}{\mathbf K}
\providecommand{\Gmat}{\mathbf G}
\providecommand{\Wmat}{\mathbf W}
\providecommand{\Dmat}{\mathbf D}
\providecommand{\mubx}{\boldsymbol\mu_x}
\providecommand{\muby}{\boldsymbol\mu_y}
\providecommand{\Cxx}{\mathbf C_{xx}}
\providecommand{\Cxy}{\mathbf C_{xy}}
\providecommand{\Cyy}{\mathbf C_{yy}}
\providecommand{\omegavec}{\boldsymbol\omega}
\providecommand{\alphavec}{\boldsymbol\alpha}
\providecommand{\nubf}{\boldsymbol\nu}
\providecommand{\etavec}{\boldsymbol\eta}

\providecommand{\Rcal}{\mathcal R}
\providecommand{\Ncal}{\mathcal N}
\providecommand{\muxm}[1]{\boldsymbol\mu_x^{(#1)}}
\providecommand{\muym}[1]{\boldsymbol\mu_y^{(#1)}}
\providecommand{\Cxxm}[1]{\mathbf C_{xx}^{(#1)}}
\providecommand{\Cxym}[1]{\mathbf C_{xy}^{(#1)}}
\providecommand{\Cyym}[1]{\mathbf C_{yy}^{(#1)}}
\providecommand{\dxym}[1]{\mathbf d_{xy}^{(#1)}}
\providecommand{\dyym}[1]{\mathbf d_{yy}^{(#1)}}
\providecommand{\fstar}[1]{\mathbf f_{#1}^{\star}}
\providecommand{\SmatD}{\mathbf S_{\rm D}}
\providecommand{\rvecD}{\mathbf r_{\rm D}}
\providecommand{\spmx}[1]{\widehat{\muxm{#1}}}
\providecommand{\spmy}[1]{\widehat{\muym{#1}}}
\providecommand{\spCxx}[1]{\widehat{\Cxxm{#1}}}
\providecommand{\spCxy}[1]{\widehat{\Cxym{#1}}}
\providecommand{\spCyy}[1]{\widehat{\Cyym{#1}}}
\providecommand{\spdxy}[1]{\widehat{\dxym{#1}}}
\providecommand{\spdyy}[1]{\widehat{\dyym{#1}}}
\providecommand{\spMSE}{\mathrm{MSE}}

\begin{document}
\maketitle

\begin{abstract}
Graph-signal estimation is commonly designed assuming that the underlying topology is known. In this paper, we consider estimation from general nonlinear measurements when the active topology is latent but is known to belong to a given family of candidate graphs. We first develop a class of topology-blind graph-spectral estimators that jointly exploit the candidate graphs, and characterize the performance loss compared to unconstrained   linear estimation. We derive the minimum-\ac{mse} multi-basis estimator and show that its optimal operator is the weighted projection of the topology-blind \ac{lmmse} operator onto the multi-graph filter space. This characterization provides exact expressions for the excess MSE and conditions for attaining the topology-blind
\ac{lmmse}. To reduce the number of parameters, we develop estimators based on shared spectral responses and low-rank representations, and consider data-dependent
estimators that infer the topology from the measurements, by \ac{map} selection or
posterior weighting (multiple-model).
We further extend the framework to networks
of different sizes by jointly optimizing a common
spectral function. We derive its MSE-optimal
coefficients and decompose the resulting estimation
loss into graph-filter approximation error and
the additional cost of sharing a spectral response
across networks. The proposed estimators admit
sample-based implementations, while the shared
spectral function can be applied to unseen graphs
without retraining its coefficients. 
Simulations on IEEE power networks demonstrate
the benefits of the proposed graph-spectral
estimators under topology uncertainty and
cross-network transfer.
\end{abstract}

\begin{IEEEkeywords}
Graph signal processing, linear estimation, topology uncertainty,
transferability, nonlinear measurement model,
power system state estimation, multiple-model estimation
\end{IEEEkeywords}

\acresetall

\vspace{-0.5cm}
\section{Introduction}
\Ac{gsp} provides a framework for processing signals
defined over networks by exploiting their graph structure~\cite{Shuman_Ortega_2013,ortega2018graph}.
A fundamental task is graph-signal estimation from noisy
or incomplete measurements \cite{7117446,kroizer2021bayesian}.
In many networked systems, the topology governs interactions
among node variables, often through nonlinear models.
Examples include coupled-oscillator and synchronization
dynamics~\cite{Dorfler_Bullo_Automatica_2014,Rodrigues_PhysRep_2016}, 
and power-system state estimation, where the topology
determines the nonlinear power-flow
equations~\cite{Giannakis_Wollenberg_2013,2021Anna}.
Most graph-based estimation methods assume a known topology.
However, switching operations, component failures, or connectivity changes may render the active topology unknown, degrading the performance of estimators designed for a fixed graph. Furthermore, physical networks often obey common laws despite differences in topology and size, creating opportunities to exploit information across networks to improve estimation and enable transfer across graph topologies and dimensions.

Motivated by these considerations, we study graph-signal
estimation under topology uncertainty and the transfer
of graph-based estimators across networks. We first
consider the setting in which the active topology is
unknown at inference time but belongs to a known
family of candidate graphs.
 In this setting, conventional graph filters designed for a single known topology cannot be directly applied without selecting a particular graph. Moreover, even when the
candidate topologies and their probabilities are known,
the optimal topology-blind linear estimator is generally not a graph filter of any individual candidate graph, i.e., it is not diagonalized by any single graph Fourier basis.
The main challenge is to exploit the spectral information of multiple candidate topologies while preserving the structural and computational advantages of graph filtering. A second challenge arises when an estimator should be shared by networks of different sizes, or applied to a network unseen in training: graph-filter coefficients defined per eigenvector do not carry over across dimensions, so the spectral response itself must be parameterized in a size-independent manner.

Graph-signal estimation has been extensively studied
under the assumption of a known graph topology.
In particular, the \ac{gsp}-linear minimum \ac{mse} estimator
in~\cite{kroizer2021bayesian} exploits graph-spectral
structure for Bayesian estimation, including nonlinear
observation models. 
This framework has been extended to graph-spectral
\ac{map} estimation~\cite{sagi2024}, widely-linear estimation
of complex-valued graph signals~\cite{Amar_Routtenberg},
and tracking
dynamic graph signals~\cite{Buchnik_GSPKalmanNet_2024}.
In our preliminary work~\cite{Zeltzer_Routtenberg_CISS},
we considered shared-response estimation of same-size
networks with the active topology known at inference time.
In contrast, the present work develops new estimators
for the entire candidate family when the active topology
is unknown, provides a detailed analysis of their
performance and properties, extends the framework
to networks of different sizes, and includes broader
comparisons with existing methods.
The effects of topology uncertainty and graph errors on \ac{gsp}
have also been investigated in graph filtering over
random graphs~\cite{Isufi2017random},
signal processing under graph perturbations
~\cite{Ceci_Barbarossa_2020,Miettinen_Vorobyov_Ollila_2021},
joint signal and graph perturbation
~\cite{Ceci2018TLS}, and robust graph-filter
identification~\cite{Rey2023robust}.
Estimation under model uncertainty has also been
addressed through multiple-model
estimation~\cite{BlomBarShalom1988,Magill1965} and robust
minimax designs~\cite{Eldar2005robust}.
However, these approaches do not explicitly exploit
the graph-spectral structure of candidate topologies.
Another line of research addresses unknown topology through graph learning, topology identification, and detection of edge disconnections from graph
signals~\cite{7990048,Dong_Vandergheynst_2016,
kalofolias2016learn,Shaked2021}.
Related methods have been developed for topology
identification in power networks
\cite{Deka_Chertkov_Backhaus_2017,
Kekatos_Giannakis_Baldick2016,2018arXiv180304506C},
joint topology and state estimation
\cite{Grotas_2019,Karimi_Natarajan_2021},
and joint graph learning and signal reconstruction
\cite{ramezani2019graph,cao2027exogenous}.
However, these approaches do not directly address the design of computationally efficient graph-spectral estimators that exploit multiple possible topologies under general nonlinear observation models.

Another related research direction concerns the
transferability of graph-based processing methods
across networks.
Graphon signal processing provides a framework
for relating graph operators across networks of
different sizes~\cite{Ruiz_Chamon_Ribeiro_GSP_2021},
while the stability and transferability of \acp{gnn} have been investigated
in~\cite{Isufi_Ribeiro,Ruiz_Chamon_Ribeiro_Transfer}.
However, the design of graph-spectral estimators
that explicitly optimize estimation performance
across networks with different topologies and sizes
remains largely unexplored.

In this paper, we study graph-signal estimation
under unknown topology and transfer across networks
of different sizes. For general, possibly nonlinear
observation models, we develop structured affine
estimators that exploit graph-spectral information.
The main contributions are as follows:
\begin{enumerate}
\item \textit{Multi-basis estimation under topology uncertainty:}
We derive the minimum-\ac{mse} multi-basis graph
filter that jointly exploits the spectral bases
of candidate topologies. We establish its
weighted-projection interpretation, characterize
its excess \ac{mse} relative to the topology-blind
\ac{lmmse} estimator, and derive conditions for
linear optimality and performance bounds.

\item \textit{Structured and topology-adaptive estimators:}
We develop lower-complexity shared-response and
low-rank estimators and characterize the \ac{mse}
loss due to their structural constraints.
We further develop data-dependent estimators
that adapt to the topology through measurement-based
\ac{map} selection or posterior weighting.

\item \textit{Transfer across network sizes:}
We develop a shared spectral-function estimator
jointly optimized across networks of different
dimensions. We derive its minimum-\ac{mse}
coefficients, decompose the estimation loss into
graph-filter approximation and cross-network sharing
errors, and provide a spectral-fitting interpretation
and an eigendecomposition-free implementation.

\item \textit{Sample-based design and evaluation:}
We derive empirical-\ac{mse} designs for all proposed
estimators and characterize their parameterization
complexity. We evaluate them on IEEE power systems
against robust minimax, expected-graph, classical
multiple-model, and deep-learning estimators, quantifying
the benefits of graph-spectral structure, topology
information, and cross-network transfer.
\end{enumerate}

The remainder of this paper is organized as follows.
Section~\ref{sec:model} introduces the problem formulation and
establishes the estimation benchmarks under topology
uncertainty. Section~\ref{sec:multibasis} develops the topology-blind
structured estimators and characterizes their
performance. Section~\ref{sec:crosssize} presents the transferable
graph-spectral estimator for networks of different
sizes. Section~\ref{sec:samples} develops sample-based implementations.
Section~\ref{sec:simulations} presents the simulations, and Section \ref{conc_subsec} concludes the paper.

\noindent\textbf{Notation:}
Boldface lowercase and uppercase letters denote vectors and matrices.
The operators $(\cdot)^T$, $(\cdot)^{-1}$, and
$(\cdot)^\dagger$ denote transpose, inverse, and
Moore--Penrose pseudoinverse.
The symbols $\odot$ and $\oslash$ denote element-wise
multiplication and division, while $\diag(\cdot)$
forms a diagonal matrix from a vector or extracts
the diagonal of a matrix.
The symbols $\mathbf{1}$ and $\Imat_N$ denote the
all-ones vector and $N\times N$ identity matrix.
The set $\mathbb S_{++}^N$ contains real symmetric
positive-definite $N\times N$ matrices.
The operators $\range(\Amat)$ and $\Ncal(\Amat)$
denote the range and null space of $\Amat$.
The symbols $\tr(\cdot)$ and $\mathbb E[\cdot]$
denote trace and expectation, while $\|\cdot\|$
and $\|\cdot\|_{\rm F}$ denote the Euclidean and
Frobenius norms.
For $\Amat\in\mathbb R^{p\times q}$ and
$\Cmat\in\mathbb S_{++}^{q}$, define
$\|\Amat\|_{\Cmat}^2
\define \tr(\Amat\Cmat\Amat^T)$.

\section{Problem Formulation and GSP Background}\label{sec:model}
In this section, we first 
provide the necessary \ac{gsp} background in Subsection~\ref{subsection_GSP}, and introduce the observation model with an
unknown topology in Subsection~\ref{subsection_problem_formulation}.  Subsection~\ref{subsec:price} formulates the estimation objective and
establishes benchmark performance measures.

\vspace{-0.25cm}
\subsection{Graph Signal Processing Background}
\label{subsection_GSP}
An undirected weighted graph $\mathcal{G}(\mathcal{V},\mathcal{E},\Wmat)$ with
$N=|\mathcal{V}|$ nodes is represented by its graph Laplacian
\begin{equation}
    \Lmat=\diag(\Wmat\onevec)-\Wmat
    =\Vmat\Lambdamat\Vmat^T,
    \label{eq:laplacian_evd}
\end{equation}
where the columns of $\Vmat$ are the Laplacian
orthogonal eigenvectors, and
$\Lambdamat=\diag(\lambda_1,\ldots,\lambda_N)$ contains the
eigenvalues, referred to as graph frequencies.  For a graph signal $\avec\in\mathbb R^N$, the \ac{gft} and its
inverse are \cite{Shuman_Ortega_2013,ortega2018graph}
\begin{equation}
    \widetilde{\avec}=\Vmat^T\avec,
    \qquad
    \avec=\Vmat\widetilde{\avec}.
    \label{eq:gft_background}
\end{equation}
A Laplacian-based linear shift-invariant graph filter is defined as
\begin{equation}
    f(\Lmat)=\Vmat f(\Lambdamat)\Vmat^T,
    \label{eq:graph_filter_background}
\end{equation}
where $f(\lambda)$ is its graph-frequency response
\cite{Shuman_Ortega_2013,ortega2018graph}. In particular, polynomial graph filters can
be written as
\begin{equation}
    f(\Lmat)=\sum_{q=0}^{Q-1}\alpha_q\Lmat^q.
    \label{eq:polynomial_filter_background}
\end{equation}
Such
representations avoid an explicit \ac{evd} and admit localized implementations.

\subsection{Signal Model with Unknown Topology}
\label{subsection_problem_formulation}
We consider $M$ candidate undirected weighted graphs on the same set of $N$ nodes,
with Laplacian matrices
\begin{equation}
    \Lmat_m=\Vmat_m\Lambdamat_m\Vmat_m^T,
    \label{eq:evd}
\end{equation}
where $\Vmat_m^T\Vmat_m=\Imat_N$ and  $\Lambdamat_m=
    \diag(\lambda_{m,1},\ldots,\lambda_{m,N})$.
We assume that the candidate Laplacian matrices have distinct eigenvalues, ordered as
$0=\lambda_{m,1}<\cdots<\lambda_{m,N}$.
The graph governing each realization is latent and is referred to as the
\emph{active topology}; its index is denoted by $m^\star$,  with
\begin{equation}
    \Pr(m^\star=m)=p_m>0,
    \qquad \sum_{m=1}^{M}p_m=1.
    \label{eq:topology_prior}
\end{equation}
We assume that the candidate family and the prior $\{\Lmat_m,p_m\}_{m=1}^{M}$ are known. We initially consider a common node set so that all candidate graphs act on
the same signal space. Section~\ref{sec:crosssize} removes this equal-dimension restriction.

Conditioned on $m^\star=m$, the measurement model is
\begin{equation}
    \yvec=\gvec(\Lmat_m,\xvec)+\wvec,
    \label{eq:nonlinear_model}
\end{equation}
where $\xvec\in\mathbb R^N$ is the random graph signal of interest, $\gvec$ is a known, generally nonlinear, topology-dependent measurement function, and $\wvec$ denotes additive measurement noise.
The model in \eqref{eq:nonlinear_model} provides a general representation for graph-based inverse and estimation problems~\cite{ortega2018graph,7117446}, and naturally includes network-dependent physical models such as power-system state estimation
\cite{2021Anna,Wang_Giannakis_Sidiropoulos2018,Monticelli1999}. The conditional distribution of $\xvec$ may also depend on $m$.  
Thus, topology uncertainty can affect both the forward model and the signal prior. 

The estimator $\hat{\xvec}$ observes $\yvec$ but not $m^\star$.   The topology-averaged \ac{mse} is
\be
\Rcal(\hat\xvec)
\define
\mathbb E\!\left[
\|\hat\xvec-\xvec\|^2
\right]
\hspace{-0.05cm}=\hspace{-0.05cm}
\sum_{m=1}^{M} p_m
\mathbb E\!\left[
\|\hat\xvec-\xvec\|^2
| m^\star=m
\right].
\label{eq:bayes_risk}
\ee
Thus, $\Rcal(\hat\xvec)$ is the Bayesian \ac{mse} under the topology prior
$\{p_m\}_{m=1}^{M}$.


For random vectors $\avec$ and $\bvec$ of appropriate dimensions, we use the notations 
$\boldsymbol\mu_a^{(m)}\define\mathbb E[\avec| m^\star=m]$ and 
\be
    \Cmat_{ab}^{(m)}
    \define
    \mathbb E\!\left[
    (\avec-\boldsymbol\mu_a^{(m)})
    (\bvec-\boldsymbol\mu_b^{(m)})^T
    | m^\star=m
    \right].
    \label{eq:conditional_moments}
\ee
Similarly, the corresponding unconditional moments under the topology prior are denoted as $\boldsymbol\mu_a
    =
    \sum_{m=1}^{M}p_m\boldsymbol\mu_a^{(m)}$
     and 
\be
    \Cmat_{ab}
    =
    \sum_{m=1}^{M}p_m
    \left[
    \Cmat_{ab}^{(m)}
    +
    (\boldsymbol\mu_a^{(m)}-\boldsymbol\mu_a)
    (\boldsymbol\mu_b^{(m)}-\boldsymbol\mu_b)^T
    \right].
    \label{eq:mixture_moments}
\ee
We assume finite second moments and $\Cyy\succ\mathbf 0$.

\subsection{Estimation Problem and Performance Benchmarks}
\label{subsec:price}


The estimator minimizing \eqref{eq:bayes_risk} over all measurable
functions of $\yvec$ is the topology-blind \ac{mmse} estimator
\begin{equation}
    \hat\xvec_{\rm U}^{\rm MMSE}
    =
    \sum_{m=1}^{M}
    \Pr(m^\star=m|\yvec)
    \mathbb E[\xvec|\yvec,m^\star=m],
    \label{eq:mmse_unknown}
\end{equation}
which follows the classical
multiple-model estimator~\cite{Magill1965}.
Since the quantities in \eqref{eq:mmse_unknown} are generally unavailable in closed form, affine estimation based on
first- and second-order statistics provides a common tractable alternative.

The optimal topology-blind affine estimator under
\eqref{eq:bayes_risk} is
\begin{equation}
    \hat\xvec_{\rm L}^\star
    =
    \mubx+\Astar(\yvec-\muby),
    \qquad
    \Astar=\Cxy\Cyy^{-1},
    \label{eq:global_lmmse}
\end{equation}
which is the topology-blind \ac{lmmse} estimator under the mixture
distribution, termed the linear minimum average MSE (LMAMSE) estimator in our preliminary work~\cite{Zeltzer_Routtenberg_CISS}.
Moreover, for any $\Amat\in\mathbb R^{N\times N}$, the \ac{lmmse}
orthogonality principle gives
\beqna
\mathbb E\!\left[
\|\mubx+\Amat(\yvec-\muby)-\xvec\|^2
\right]
=
\Rcal(\hat\xvec_{\rm L}^\star)+
\|\Amat-\Astar \|_{\Cyy}^2,
\label{eq:lmmse_projection}
\eeqna
where $\hat\xvec_{\rm L}^\star$ and $\Astar $ are given in \eqref{eq:global_lmmse}. The term in \eqref{eq:lmmse_projection} evaluates the \ac{mse} gap between a general linear estimator and the \ac{lmmse} estimator.

A second benchmark assumes knowledge of the active topology. If
$m^\star$ were revealed, the corresponding oracle \ac{lmmse} estimator would be
\beqna
    \hat\xvec_{\rm O}^{(m^\star)}(\yvec)
    & = &
    \muxm{m^\star}
    +
    \Amat_{\rm O}^{(m^\star)}
    (\yvec-\muym{m^\star}),
    \nonumber\\
    \Amat_{\rm O}^{(m^\star)}
    & \define &
    \Cxym{m^\star}
    \big(\Cyym{m^\star}\big)^{-1},
    \label{eq:conditional_lmmse}
\eeqna
where  we assume $\Cyym{m}\succ\zerovec$, $\forall m$.
By the conditional \ac{lmmse} orthogonality principle, the \ac{mse} gap between
the topology-blind \ac{lmmse} estimator and the oracle benchmark satisfies
\beqna
\Rcal(\hat\xvec_{\rm L}^\star)
-
\sum_{m=1}^{M}p_m
\mathbb E \big[
\|\hat\xvec_{\rm O}^{(m)}(\yvec)-\xvec\|^2
|\, m^\star=m
\big]\hspace{1.25cm}
\nonumber\\
=
\sum_{m=1}^{M}p_m
\mathbb E\!\big[
\|\hat\xvec_{\rm L}^\star(\yvec)
-\hat\xvec_{\rm O}^{(m)}(\yvec)\|^2
|\, m^\star=m
\big ]
\geq0.
\label{eq:linear_price}
\eeqna
Thus, \eqref{eq:linear_price} quantifies the \ac{mse} cost of not observing the active topology within affine estimation. In the following, $\Astar$ serves as the unstructured topology blind \ac{lmmse} benchmark.

\section{Topology-Blind Structured Graph-Signal Estimation}\label{sec:multibasis}
In this section, we introduce the multi-basis GSP-\ac{lmmse} estimator in Subsection~\ref{subsec:dictionary} and characterize its performance in Subsection~\ref{subsec:projection}. Subsection~\ref{subsec:shared_blind} develops lower-dimensional shared-response and low-rank variants. Finally, Subsection~\ref{subsec:topology_selection} considers complementary data-dependent estimators that select or weight the candidate topologies from each observation. 

\subsection{Multi-Basis GSP-\ac{lmmse} Estimator}
\label{subsec:dictionary}
For a known graph, the GSP-\ac{lmmse} estimator in
\cite{kroizer2021bayesian} is diagonal in the graph Fourier basis. Under topology uncertainty, however, no single candidate basis is known to be active. This motivates a topology-blind estimator that combines the spectral projectors of all candidate graphs.  

Using the $M$ candidate graph bases in \eqref{eq:evd}, consider the
topology-blind estimator
\begin{equation}
    \hat\xvec_{\rm MB}(\yvec)
    =
    \mubx+
    \sum_{m=1}^{M}
    \Vmat_m\diag(\hvec_m)\Vmat_m^T(\yvec-\muby),
    \label{eq:MB_estimator}
\end{equation}
where $\hvec_m\in\mathbb R^N$ is the spectral response associated with
the $m$th candidate graph. Unlike the data-dependent topology-selection approach developed in
Subsection~\ref{subsec:topology_selection},
\eqref{eq:MB_estimator} jointly combines all candidate graph bases
without committing to a single topology at inference.

Define $\SmatD\in\mathbb R^{MN\times MN}$ blockwise as
\begin{equation}
    [\SmatD]_{m\ell}
    \define
    (\Vmat_m^T\Cyy\Vmat_\ell)
    \odot
    (\Vmat_m^T\Vmat_\ell),
    ~~~ m,\ell=1,\ldots,M,
    \label{eq:S_blocks}
\end{equation}
and
\begin{equation}
    \rvecD
    \define[\diag(\Vmat_1^T\Cxy\Vmat_1)^T,\ldots,
\diag(\Vmat_M^T\Cxy\Vmat_M)^T]^T.
    \label{eq:r_blocks}
\end{equation}

\begin{theorem}[Multi-basis GSP-\ac{lmmse} estimator]
\label{thm:MB}
Among the topology-blind estimators \hspace{-0.02cm}\eqref{eq:MB_estimator}, the
minimum-\ac{mse} estimator is
\begin{equation}
    \hat\xvec_{\rm MB}^\star(\yvec)
    =
    \mubx+
    \sum_{m=1}^{M}
    \Vmat_m\diag(\hvec_m^\star)\Vmat_m^T
    (\yvec-\muby),
    \label{eq:opt_est}
\end{equation}
where the minimum-norm optimal coefficient vector is
\begin{equation}
    \etavec^\star
    \define
    [(\hvec_1^\star)^T,\ldots,(\hvec_M^\star)^T]^T
    =
    \SmatD^\dagger\rvecD.
    \label{eq:eta_solution}
\end{equation}
The corresponding minimum \ac{mse} is
\begin{equation}
    \Rcal\big(\hat\xvec_{\rm MB}^\star\big)
    =
    \tr(\Cxx)
    -
    \rvecD^T\SmatD^\dagger\rvecD.
    \label{eq:MB_minrisk}
\end{equation}
\end{theorem}
\begin{proof}
The proof appears in Appendix~\ref{AppA}.
\end{proof}

The cross-basis blocks in \eqref{eq:S_blocks} account jointly for the
measurement covariance and spectral overlap between candidate graph
bases. Thus, \eqref{eq:MB_estimator} is jointly optimized across the
candidate graphs rather than obtained by independently designing and
combining single-graph filters.
In addition, the optimal estimator is unique, although its coefficient representation in \eqref{eq:eta_solution}
is not. In particular, for $M\geq2$, any $\{c_m\}_{m=1}^{M}$ satisfying
$\sum_{m=1}^{M}c_m=0$ can be used to replace
$\hvec_m$ by $\hvec_m+c_m\onevec$ without changing
\eqref{eq:MB_estimator}, since each graph basis contains the identity
component. Consequently,
\begin{equation}
    \dim\Ncal(\SmatD)\geq M-1.
    \label{eq:SD_nullity}
\end{equation}
Thus, \eqref{eq:eta_solution} selects the minimum-norm coefficient
representation, while $\Cyy\succ\zerovec$ guarantees uniqueness of the
resulting optimal estimator.

\subsection{Performance Characterization}
\label{subsec:projection}
An immediate consequence of the multi-basis construction is that it contains every topology-blind GSP estimator restricted to a single
candidate basis: for any $m$, set
$\hvec_\ell=\delta[\ell-m]\hvec$, $\ell=1,\ldots,M$. Thus,
\beqna
\Rcal(\hat\xvec_{\rm MB}^\star)
\leq\hspace{5.5cm}
\nonumber\\
\min_{m=1,\ldots,M}\;
\min_{\hvec}
\Rcal\!\left(
\mubx+
\Vmat_m\diag(\hvec)\Vmat_m^T
(\yvec-\muby)
\right).
\label{eq:single_basis_dominance}
\eeqna
Thus, at the population level, jointly exploiting the candidate graph
bases cannot perform worse than the best topology-blind single-basis
GSP estimator. 
The following result characterizes the performance loss induced by the
multi-basis graph-filter structure relative to the unconstrained
topology-blind \ac{lmmse} estimator and describes the achievability conditions. 
\begin{theorem}[MSE gap of the multi-basis GSP estimator]
\label{thm:projection}
The optimal multi-basis estimator satisfies
\beqna
    \Rcal\big(\hat\xvec_{\rm MB}^\star\big)
    -
    \Rcal\big(\hat\xvec_{\rm L}^\star\big)\hspace{4.5cm}
    \nonumber\\
    =
    \Big\|
    \sum_{m=1}^{M}
    \Vmat_m\diag(\hvec_m^\star)\Vmat_m^T
    -\Astar
    \Big\|_{\Cyy}^2.
    \label{eq:MB_gap}
\eeqna
Its induced operator is the $\Cyy$-weighted orthogonal projection of
$\Astar$ onto the multi-basis operator space. Consequently, it attains
the topology-blind \ac{lmmse} {\em{iff}}
\begin{equation}
\label{iff_A_star}
    \Astar
    =
    \sum_{m=1}^{M}
    \Vmat_m\diag(\hvec_m)\Vmat_m^T
\end{equation}
for some $\{\hvec_m\}_{m=1}^{M}$.
\end{theorem}
\begin{proof}
Applying \eqref{eq:lmmse_projection} to the multi-basis class \eqref{eq:MB_estimator} gives
\eqref{eq:MB_gap} with $\hvec_m$. Minimizing the \ac{mse} is therefore equivalent to
minimizing the $\Cyy$-weighted distance from $\Astar$ over the
admissible multi-basis operators, which proves the projection statement.
If \eqref{iff_A_star} holds, then $\Astar$ belongs to this class and its
projection is itself, so the \ac{mse} gap is zero. Conversely, if the gap is
zero, $\Cyy\succ\zerovec$ implies that the optimal multi-basis operator
equals $\Astar$, and hence \eqref{iff_A_star} holds.
\end{proof}
Equation~\eqref{eq:MB_gap} shows that the performance loss induced by
the multi-basis constraint is the squared $\Cyy$-weighted
distance between the topology-blind \ac{lmmse} operator and the multi-basis
operator space. The loss therefore vanishes precisely when $\Astar$
admits the representation in \eqref{iff_A_star}. In particular, since every 
multi-basis operator is symmetric, symmetry of $\Astar$ is necessary
for attaining the topology-blind \ac{lmmse}.

\begin{corollary}[Sufficient condition for attaining the \ac{lmmse}]
\label{cor:linear_opt}
Suppose that 1) the conditional means are topology independent;
2) the candidate topologies share an eigenbasis $\Vmat$; and
3) $\Vmat^T\Cxym{m}\Vmat$ and
$\Vmat^T\Cyym{m}\Vmat$ are diagonal for every $m$.
Then the multi-basis GSP-\ac{lmmse} estimator attains the topology-blind
\ac{lmmse}.
\end{corollary}

\begin{proof}
The proof appears in Appendix~\ref{App_b}.
\end{proof}
In this case, all candidate graph-filter classes share the same
eigenbasis, so the multi-basis operator class reduces to a single
common-basis graph-filter class. For symmetric graph Laplacians, such
a common eigenbasis exists {\em{iff}} the Laplacians commute
pairwise~\cite{HornJohnsonMatrixAnalysis,Bronstein2014Making}.
In general, however, diagonal conditional moments in different graph bases do not imply this property, since the mixture moments need not share any of these bases. The following corollary considers a complementary setting in which the
oracle LMMSE operator associated with each candidate topology is itself
a graph filter.

\begin{corollary}[Effect of covariance heterogeneity]
\label{cor:heterogeneity}
Assume that the conditional means are topology independent and that the oracle LMMSE operator corresponding to
$m^\star=m$ is a graph filter \ac{wrt} $\Lmat_m$, $\forall m=1,\ldots,M$,  i.e.,
\begin{equation}
    \Amat_{\rm O}^{(m)}
    =
    \Vmat_m\diag(\hvec_m^{\rm O})\Vmat_m^T
    \label{eq:topology_lmmse_graph_filter}
\end{equation}
for some $\hvec_m^{\rm O}\in\mathbb R^N$. Then
\beqna
\Rcal(\hat\xvec_{\rm MB}^\star)
-\Rcal(\hat\xvec_{\rm L}^\star)
\hspace{4.5cm}
\nonumber\\
\quad\leq
\Big\|
\sum_{m=1}^{M}p_m
\Amat_{\rm O}^{(m)}
\big(\Cyym{m}-\Cyy\big)
\Cyy^{-1/2}
\Big\|_{\rm F}^{2}.
\label{eq:heterogeneity_bound}
\eeqna
\end{corollary}
\begin{proof}
The proof appears in Appendix~\ref{App_B2}.
\end{proof}

Corollary~\ref{cor:heterogeneity} shows that the excess
\ac{mse} is bounded by covariance heterogeneity across
topologies. In particular, the bound is small when the
covariance-weighted deviations
$\Amat_{\rm O}^{(m)}(\Cyym{m}-\Cyy)$
are small or cancel in the weighted sum.

The unconstrained \ac{lmmse} in \eqref{eq:global_lmmse} is the affine-optimal
benchmark, but it uses a general $N\times N$ mapping and ignores graph
structure. The multi-basis GSP-\ac{lmmse} estimator exploits the candidate
graph spectra and provides a lower-dimensional design. In finite samples,
however,  its flexibility may increase estimation variance and
numerical sensitivity. This motivates
the lower-complexity structured estimators.

\vspace{-0.25cm}
\subsection{Shared-Response Estimator and Low-Rank Extension}
\label{subsec:shared_blind}
The multi-basis GSP-\ac{lmmse} estimator in \eqref{eq:opt_est} uses
$MN$ spectral coefficients. To reduce this complexity and exploit structural similarity among the
candidate graphs, we impose a separable structure on the
multi-basis responses,
\begin{equation}
    \hvec_m=\omega_m\hvec,\qquad m=1,\ldots,M,
    \label{eq:shared_response_constraint}
\end{equation}
where $\hvec\in\mathbb R^N$ is a common spectral response and
$\{\omega_m\}_{m=1}^M$ are graph-dependent weights. Substituting
\eqref{eq:shared_response_constraint} into
\eqref{eq:MB_estimator} gives
\begin{equation}
    \hat\xvec_{\rm SR}(\yvec)
    =
    \mubx+
    \sum\nolimits_{m=1}^{M}\omega_m
    \Vmat_m\diag(\hvec)\Vmat_m^T(\yvec-\muby),
    \label{eq:shared_estimator}
\end{equation}
where 
$\omega_m\geq0$, and $\sum_{m=1}^{M}\omega_m=1$. 
This restriction reduces the number of spectral-response parameters
from $MN$ to $N$ for fixed weights.
The weights are design parameters: $\omega_m=p_m$ is prior-weighted fusion, while $\omega_m\hspace{-0.1cm}=\hspace{-0.1cm}\delta[m\hspace{-0.1cm}-\hspace{-0.1cm}k]$ recovers a topology-blind single-basis GSP estimator of the candidate $k$.

\begin{proposition}[Optimal shared-response estimator]
\label{prop:shared_blind}
For fixed weights $\omegavec$, the minimum-norm \ac{mse}-optimal response
within the shared-response class \eqref{eq:shared_estimator} is 
\beqna
\hvec^\star
=
\Big(
\sum\nolimits_{m=1}^{M}\sum\nolimits_{\ell=1}^{M}
\omega_m\omega_\ell
(\Vmat_m^T\Cyy\Vmat_\ell)
\odot
(\Vmat_m^T\Vmat_\ell)
\Big)^\dagger
\nonumber\\
\quad\times
\sum\nolimits_{m=1}^{M}\omega_m
\diag(\Vmat_m^T\Cxy\Vmat_m).
\label{eq:shared_solution}
\eeqna
\end{proposition}

\begin{proof}
Substituting $\hvec_m=\omega_m\hvec$ into the \ac{mse} criterion \eqref{eq:bayes_risk} (or equivalently, into \eqref{eq:MB_risk} from Appendix~\ref{AppA}) yields a convex quadratic function of $\hvec$. Its normal equations give \eqref{eq:shared_solution}, whose Moore--Penrose solution is the minimum-norm minimizer.
\end{proof}

For fixed
$\omegavec$, the shared-response operator space is a linear subspace
of the multi-basis space. Hence, by the Pythagorean identity for orthogonal projections \cite{luenberger1969optimization} and Theorem~\ref{thm:projection},
the \ac{mse} of the shared-response restriction is
\beqna
\Rcal(\hat\xvec_{\rm SR}^\star)
-\Rcal(\hat\xvec_{\rm MB}^\star)\hspace{4.75cm}
\nonumber\\
=
\Big\|
\sum\nolimits_{m=1}^{M}
\Vmat_m
\diag(\omega_m\hvec^\star-\hvec_m^\star)
\Vmat_m^T
\Big\|_{\Cyy}^{2}.
\label{eq:shared_response_gap}
\eeqna
where $\hvec_m^\star$ and $\hvec^\star$ are given by
\eqref{eq:eta_solution} and \eqref{eq:shared_solution}, respectively.

\begin{remark}[Low-rank extension]
\label{rem:lowrank}
The shared-response model can be generalized by allowing a small number of spectral responses across the candidate graphs. Specifically, consider multi-basis coefficients admitting
\begin{equation}
    \Hmat=[\hvec_1,\ldots,\hvec_M]
    =\Umat\Bmat^T,
    \qquad \rank(\Hmat)\leq r,
    \label{eq:lowrank_factor}
\end{equation}
where the columns of $\Umat$ define $r$ shared spectral templates and the rows of $\Bmat$ contain their graph-dependent weights. This factorization uses $r(N+M-r)$ nominal degrees of freedom, compared with $MN$ for the unrestricted multi-basis model.

The resulting rank-$r$ optimization (minimum MSE for $\Hmat$ as in \eqref{eq:lowrank_factor}) is nonconvex
but can be addressed by alternating least squares.
For $r=1$, it reduces to the shared-response model
in Proposition~\ref{prop:shared_blind} when the
graph-dependent weights are fixed.

\end{remark}

\vspace{-0.25cm}
\subsection{Topology-Adaptive GSP Estimation}
\label{subsec:topology_selection}
The estimators in Subsection~\ref{subsec:shared_blind} fuse information across all candidate graphs. In contrast, the following approach first infers the topology or weights the candidates based on the observation and then estimates the graph signal.

Specifically, for each candidate topology \(m\), we consider the GSP affine estimator
\begin{equation}
    \hat\xvec_{\rm GSP}^{(m)}(\yvec)
    =
    \muxm{m}
    +
    \Vmat_m
    \diag(\hvec_{{\rm GSP},m}^{\star})
    \Vmat_m^T
    (\yvec-\muym{m}),
    \label{eq:topology_specific_gsp}
\end{equation}
where
\begin{equation}
    \hvec_{{\rm GSP},m}^{\star}
    =
    \dxym{m}\oslash\dyym{m},
    \label{eq:topology_specific_response}
\end{equation}
in which
\begin{subequations}
\label{eq:dxy_dyy}
\begin{align}
\dxym{m}
&\define
\diag\!\left(\Vmat_m^T\Cxym{m}\Vmat_m\right),
\label{eq:dxy_m}\\
\dyym{m}
&\define
\diag\!\left(\Vmat_m^T\Cyym{m}\Vmat_m\right),
\label{eq:dyy_m}
\end{align}
\end{subequations}
Since $\Cyym{m}\succ\zerovec$, all entries of $\dyym{m}$ are
strictly positive.
For a known topology $m$, the estimator  \eqref{eq:topology_specific_gsp}-\eqref{eq:dxy_dyy} is the minimum-\ac{mse}
estimator within the class of affine graph filters diagonalized by \(\Vmat_m\), which is the
GSP-LMMSE estimator of \cite{kroizer2021bayesian}.

Alternatively, one may use a common spectral response shared across all
candidate topologies,
\begin{equation}
    \hat\xvec_{\rm SR}^{(m)}(\yvec)
    =
    \muxm{m}
    +
    \Vmat_m\diag(\hvec_{\rm SR}^\star)\Vmat_m^T
    (\yvec-\muym{m}),
    \label{eq:topology_aware_estimator}
\end{equation}
with
\begin{equation}
    \hvec_{\rm SR}^\star
    =
    \left(\sum\nolimits_{m=1}^{M}p_m\dxym{m}\right)
    \oslash
    \left(\sum\nolimits_{m=1}^{M}p_m\dyym{m}\right).
    \label{eq:topology_aware_shared}
\end{equation}
Similar to the derivation of the shared-response estimator in
Section~\ref{subsec:shared_blind}, \(\hvec_{\rm SR}^\star\) is the minimum-\ac{mse} solution
within the class of graph filters constrained to share a common spectral
response. Thus,
\(\hvec_{{\rm GSP},m}^{\star}\) adapts the spectral response to each topology,
whereas \(\hvec_{\rm SR}^\star\) couples all candidate graphs through a
single response. As in \eqref{eq:shared_estimator}, topology-dependent
weights \(\{\omega_m\}\) can also be incorporated into the shared-response
design.


Let \(\hat m(\yvec)\in\{1,\ldots,M\}\) denote a topology-selection rule. The topology-selective GSP-LMMSE and SR-GSP-LMMSE estimators are obtained by replacing 
\(m\) in \eqref{eq:topology_specific_gsp} and
\eqref{eq:topology_aware_estimator} with the selected topology
\(\hat m(\yvec)\).
When the topology-conditioned likelihood is available, the topology
can be estimated by the \ac{map} rule
\begin{equation}
    \hat m_{\rm MAP}(\yvec)
    =
    \arg\max_{m=1,\ldots,M}
    \left\{
    p_m\,p(\yvec\mid m^\star=m)
    \right\}.
    \label{eq:map_topology_selector}
\end{equation}
This corresponds to the hard-decision counterpart of the posterior weighting in the topology-blind MMSE estimator~\eqref{eq:mmse_unknown}.
When the likelihood in \eqref{eq:map_topology_selector} is difficult to
evaluate, the forward model in \eqref{eq:nonlinear_model} suggests the
residual-based rule
\begin{equation}
    \hat m_{\rm res}(\yvec)
    =
    \arg \min_{m=1,\ldots,M}
    \left\|
    \yvec-
    \gvec(
    \Lmat_m,
    \hat\xvec_{\rm GSP}^{(m)}(\yvec)
    )
    \right\|^2.
    \label{eq:residual_topology_selector}
\end{equation}
Unlike \eqref{eq:map_topology_selector}, this model-consistency criterion requires no
explicit probabilistic model for the posterior.
Replacing the hard decision by the posterior weights
$\pi_m(\mathbf{y})\hspace{-0.1cm}\propto \hspace{-0.1cm}p_m\,p(\mathbf{y}|m^\star=m)$ gives the
multiple-model GSP estimator
\begin{equation}
\hat{\mathbf{x}}_{\rm MM}(\mathbf{y})=\sum_{m=1}^{M}\pi_m(\mathbf{y})\,
\hat{\mathbf{x}}^{(m)}_{\rm GSP}(\mathbf{y}),
\label{eq:40}
\end{equation}
which approximates the MMSE estimator~\eqref{eq:mmse_unknown} by replacing each
conditional mean  with the
affine graph-filter branch~\eqref{eq:topology_specific_gsp}. 

Although \eqref{eq:topology_specific_gsp} and
\eqref{eq:topology_aware_estimator} are affine, the topology-selective estimators are generally nonlinear in $\yvec$ because $\hat m(\yvec)$ is data dependent. Their MSE is
\beqna
    \Rcal(\hat\xvec_{\rm TS})
    =\hspace{6.5cm}\nonumber\\
    \sum_{m=1}^{M}p_m
    \sum_{\ell=1}^{M}
    \mathbb E\!\left[
    \|\hat\xvec_{\rm GSP}^{(\ell)}(\yvec)-\xvec\|^2
    \delta[\hat m(\yvec)-\ell]
   | m^\star=m
    \right],
    \label{eq:topology_selection_risk}
\eeqna
which depends on both the topology-specific estimators and
the induced selection regions. For
$\hat m(\yvec)=m^\star$ a.s., 
the
topology-specific GSP-LMMSE estimator \cite{kroizer2021bayesian} is obtained by setting $m=m^\star$
in \eqref{eq:topology_specific_gsp}.
However, in general, there is no universal \ac{mse} ordering relative to the multi-basis estimator.

\begin{remark}[Relation to shared-response GSP estimation]
\label{rem:shared_response}
Our preliminary work \cite{Zeltzer_Routtenberg_CISS} considered a
complementary setting in which the active topology \(m^\star\) is known at
inference. The corresponding Shared-Response GSP estimator is
obtained from \eqref{eq:topology_aware_estimator} by setting \(m=m^\star\).
Both the SR-GSP and the shared-response estimators employ a common spectral response across topologies. However, the SR-GSP estimator uses the revealed topology, while the shared-response
estimator in \eqref{eq:shared_estimator} is topology-blind and fuses the candidate graph bases.
\end{remark}

\setlength{\dbltextfloatsep}{6pt}   
\begin{table*}[t]
\centering
\caption{Summary of the estimators considered in Section~\ref{sec:multibasis}.}
\label{tab:estimators_summary}
\vspace{-0.15cm}
\footnotesize
\setlength{\tabcolsep}{3pt}
\begin{tabular}{|l|c|l|l|l|l|}
\hline
Estimator & Eq. & Role & Topology info. & Structure & Design criterion \\
\hline
Topology-blind LMMSE
& \eqref{eq:global_lmmse}
& Benchmark
& None
& Affine
& Global MSE \\

Multi-basis GSP-LMMSE
& \eqref{eq:opt_est}
& Proposed
& Candidate set
& Multi-basis
& Global MSE \\

Shared-response GSP-LMMSE
& \eqref{eq:shared_estimator}
& Proposed
& Candidate set
& Shared response
& Global MSE \\

Topology-selective GSP-LMMSE
& \eqref{eq:topology_specific_gsp}, $m=\hat m$
& Proposed
& Estimated
& Topology-specific
& Conditional MSE \\

Multiple-model GSP-LMMSE
& \eqref{eq:40}
& Proposed
& Estimated (posterior)
& Topology-specific, weighted
& Conditional MSE \\

Topology-selective SR-GSP-LMMSE
& \eqref{eq:topology_aware_estimator}, $m=\hat m$
& Proposed
& Estimated
& Shared response
& Topology-averaged MSE \\

Oracle GSP-LMMSE
& \eqref{eq:topology_specific_gsp}, $m=m^\star$
& Oracle
& True topology
& Topology-specific
& Conditional MSE \\

Oracle SR-GSP-LMMSE
& \eqref{eq:topology_aware_estimator}, $m=m^\star$
& Prior work
& True topology
& Shared response
& Topology-averaged MSE \\
\hline
\end{tabular}
\end{table*}

\begin{remark}[Connection to graphical Wiener filtering]
\label{rem:wiener}
Consider the linear special case of \eqref{eq:nonlinear_model},
conditioned on topology $m$,
\begin{equation}
    \yvec
    =
    \Vmat_m\diag(\boldsymbol\phi_m)\Vmat_m^T\xvec+\wvec,
    \label{eq:wiener_model}
\end{equation}
where $\xvec$ and $\wvec$ are uncorrelated and
graph wide-sense stationary \ac{wrt} $\Lmat_m$, with
graph power spectra $\svec_{x,m}$ and $\svec_{w,m}$.
Then, \eqref{eq:dxy_dyy} yields
$\dxym{m}=\svec_{x,m}\odot\boldsymbol\phi_m$ and
$\dyym{m}=\svec_{x,m}\odot\boldsymbol\phi_m^{\odot 2}
+\svec_{w,m}$.
Hence, the optimal shared response in
\eqref{eq:topology_aware_shared} reduces to
\begin{equation}
[\hvec_{\rm SR}^{\star}]_n
=
\frac{
\sum_{m=1}^{M}p_m
[\svec_{x,m}]_n[\boldsymbol\phi_m]_n
}{
\sum_{m=1}^{M}p_m
\left(
[\svec_{x,m}]_n[\boldsymbol\phi_m]_n^2
+[\svec_{w,m}]_n
\right)
},
\label{eq:wiener_shared_response}
\end{equation}
$n=1,\ldots,N$.
This is the shared-response counterpart of graphical
Wiener filtering~\cite{kroizer2021bayesian,7891646}.
\end{remark}


Table~\ref{tab:estimators_summary} summarizes the estimators
developed in this section, distinguishing the topology-blind
designs from topology-aware benchmarks and topology-selection
alternatives.

\section{Transferable Graph-Signal Estimation Across Network Sizes}
\label{sec:crosssize}
\vspace{-0.05cm}
The shared-response estimator in Section~\ref{sec:multibasis} requires candidate
graphs of the same dimension. In this section, we instead parameterize the response as a common function of graph frequency, enabling joint estimation across networks of different sizes.
Subsection \ref{subsec:crosssize} introduces the shared-function representation and its cross-network design. Then, Subsection \ref{subsec:prop} provides its interpretation
and implementation.

\vspace{-0.25cm}
\subsection{Shared Spectral Function and Cross-Network Design}
\label{subsec:crosssize}

Consider \(J\) network groups, where group \(j\) contains candidate graphs
\(\{\Lmat_{j,k}\}_{k\in\Ccal_j}\) with \(N_j\) nodes. The group index is
known at inference and determines the relevant candidate family. Let
\(\muxj{j}\), \(\muyj{j}\), \(\Cxyj{j}\), and
\(\Cyyj{j}\succ\zerovec\) denote the corresponding mixture moments, defined
as in \eqref{eq:conditional_moments}--\eqref{eq:mixture_moments}.

To enable transfer across graph sizes, we parameterize the spectral response
using a fixed dictionary. For topology weights \(\omega_{j,k}\geq0\), with
\(\sum_{k\in\Ccal_j}\omega_{j,k}=1\), consider
\begin{equation}
    \hat\xvec_j(\yvec;\alphavec)
    =
    \muxj{j}
    +
    \sum_{k\in\Ccal_j}
    \omega_{j,k}
    h(\Lmat_{j,k};\alphavec)
    (\yvec-\muyj{j}),
    \label{eq:crosssize_estimator}
\end{equation}
where
\begin{equation}
    h(\Lmat;\alphavec)
    =
    \Vmat
    \left(
        \sum_{q=0}^{Q-1}\alpha_q\phi_q(\Lambdamat)
    \right)
    \Vmat^T,
    \label{eq:dictionary_graph_filter}
\end{equation}
and \(\{\phi_q\}_{q=0}^{Q-1}\) is a predefined spectral dictionary, with
\(\phi_q(\Lambdamat)\) obtained by applying \(\phi_q\) elementwise to the
eigenvalues. Since \(\alphavec\in\mathbb R^Q\) is independent of the graph
dimension, the same spectral response can be applied across network sizes.
For models with an additive-reference ambiguity \cite{routtenberg2020}, we
assume that a reference is fixed prior to estimation.

We design the common coefficients by minimizing the average per-node MSE
\begin{equation}
    \Rcal_{\rm XS}(\alphavec)
    \define
    \frac{1}{J}\sum_{j=1}^{J}\frac{1}{N_j}
    \mathbb E_j\!\left[
        \|\hat\xvec_j(\yvec;\alphavec)-\xvec\|^2
    \right].
    \label{eq:crosssize_risk}
\end{equation}
For $J=1$ (and hence $N_1=N$), \eqref{eq:crosssize_risk} reduces,
up to the normalization factor $1/N$, to the MSE in
\eqref{eq:bayes_risk}. The normalization by $N_j$ enables comparison
across networks of different sizes.

For the following theorem, we define the group-wise dictionary operators for $q=0,\ldots,Q-1$, as 
\begin{equation}
    \Hmat_{j,q}
    \define
    \sum_{k\in\Ccal_j}
    \omega_{j,k}
    \Vmat_{j,k}
    \phi_q(\Lambdamat_{j,k})
    \Vmat_{j,k}^T,
    \label{eq:Hjq}
\end{equation}
Further, define $\mathbf S_{\rm XS}\in\mathbb R^{Q\times Q}$ and
$\mathbf r_{\rm XS}\in\mathbb R^Q$ by
\begin{subequations}
\label{eq:XS_matrices}
\begin{align}
[\mathbf S_{\rm XS}]_{q+1,r+1}
&\define
\sum\nolimits_{j=1}^{J}
\frac{
   1
}{N_j}  \tr(\Hmat_{j,q}\Cyyj{j}\Hmat_{j,r}^T),
\label{eq:SXS}\\
[\mathbf r_{\rm XS}]_{q+1}
&\define
\sum\nolimits_{j=1}^{J}
\frac{
    1
}{N_j}\tr(\Hmat_{j,q}\Cxyj{j}^{T}),
\label{eq:rXS}
\end{align}
\end{subequations}
for $q,r=0,\ldots,Q-1$.

\begin{theorem}[Cross-network GSP-\ac{lmmse} estimator]
\label{thm:crosssize}
For any fixed dictionary $\{\phi_q\}_{q=0}^{Q-1}$, the minimum-norm
coefficient vector minimizing \eqref{eq:crosssize_risk} is
\begin{equation}
    \alphavec_{\rm XS}^\star
    =
    \mathbf S_{\rm XS}^{\dagger}\mathbf r_{\rm XS}.
    \label{eq:crosssize_alpha}
\end{equation}
The corresponding minimum average per-node MSE is
\begin{equation}
    \Rcal_{\rm XS}(\alphavec_{\rm XS}^\star)
    =
    \frac{1}{J}\sum_{j=1}^{J}
    \frac{\tr(\Cxxj{j})}{N_j}
    -
    \frac{1}{J}
    \mathbf r_{\rm XS}^{T}
    \mathbf S_{\rm XS}^{\dagger}
    \mathbf r_{\rm XS}.
    \label{eq:MB_minrisk_cross}
\end{equation}
\end{theorem}

\begin{proof}
Using \eqref{eq:Hjq}, the linear operator in
\eqref{eq:crosssize_estimator} is
$\sum_{q=0}^{Q-1}\alpha_q\Hmat_{j,q}$. Hence,
\eqref{eq:crosssize_risk} can be written as
\beqna
    \Rcal_{\rm XS}(\alphavec) \hspace{-0.05cm}=\hspace{-0.05cm}
    \frac{1}{J}\sum_{j=1}^{J}
    \frac{\tr(\Cxxj{j})}{N_j}
    \hspace{-0.05cm}+\hspace{-0.05cm}
    \frac{1}{J}\hspace{-0.05cm}
    \left(
        \alphavec^T\mathbf S_{\rm XS}\alphavec
        -
        2\alphavec^T\mathbf r_{\rm XS}
    \right)\hspace{-0.05cm}.
    \label{eq:crosssize_quadratic}
\eeqna
The normal equations are thus
$\mathbf S_{\rm XS}\alphavec=\mathbf r_{\rm XS}$, whose
minimum-norm is \eqref{eq:crosssize_alpha}.
Substituting \eqref{eq:crosssize_alpha} into \eqref{eq:crosssize_quadratic}
gives \eqref{eq:MB_minrisk_cross}.
\end{proof}

Theorem~\ref{thm:crosssize} can be viewed as the cross-network counterpart of Theorem~\ref{thm:MB}: the fixed-size multi-basis coefficients are replaced by a common spectral function that is jointly optimized across networks of different dimensions.
We next characterize the performance loss induced by the shared cross-network spectral function. Applying the \ac{lmmse} orthogonality principle group-wise, similarly to
Theorem~\ref{thm:projection}, gives
\beqna
    \Rcal_{\rm XS}(\alphavec)
    =
    \Rcal_{{\rm L},{\rm XS}}^\star\hspace{5cm}
    \nonumber\\
    +
    \frac{1}{J}\sum\nolimits_{j=1}^{J}\frac{1}{N_j}
    \left\|
        \sum\nolimits_{q=0}^{Q-1}\alpha_q\Hmat_{j,q}
        -
        \Cxyj{j}\Cyyj{j}^{-1}
    \right\|_{\Cyyj{j}}^2,
    \label{eq:crosssize_lmmse_decomp}
\eeqna
where the average per-node MSE of the group-wise topology-blind
\ac{lmmse} estimators is
\begin{equation}
   \hspace{-0.15cm} \Rcal_{{\rm L},{\rm XS}}^\star
    \define
    \frac{1}{J}\sum_{j=1}^{J}\frac{1}{N_j}
    \left(
        \tr(\Cxxj{j})
        -
        \tr(
            \Cxyj{j}\Cyyj{j}^{-1}\Cxyj{j}^T
        )
    \right)
    \label{eq:crosssize_lmmse_risk}
\end{equation}
 Thus, the excess MSE in
\eqref{eq:crosssize_lmmse_decomp} is the average weighted
distance between the operators induced by the common spectral function
and the group-wise topology-blind \ac{lmmse} operators.

Furthermore, let
\(\alphavec_j^\star=[\alpha_{j,0}^\star,\ldots,
\alpha_{j,Q-1}^\star]^T\) denote the optimal coefficients when group \(j\)
is designed separately. The corresponding operator
\(\sum_{q=0}^{Q-1}\alpha_{j,q}^\star\Hmat_{j,q}\) is the
\(\Cyyj{j}\)-weighted projection of
\(\Cxyj{j}\Cyyj{j}^{-1}\) onto the group-\(j\) graph-filter space.
Therefore, 
the Pythagorean identity gives
\beqna
\Rcal_{\rm XS}(\alphavec_{\rm XS}^\star)
-\Rcal_{{\rm L},{\rm XS}}^\star
\hspace{5.25cm}
\nonumber\\
=
\frac{1}{J}\sum\nolimits_{j=1}^{J}\frac{1}{N_j}
\Big\|
    \sum\nolimits_{q=0}^{Q-1}\alpha_{j,q}^\star\Hmat_{j,q}
    -
    \Cxyj{j}\Cyyj{j}^{-1}
\Big\|_{\Cyyj{j}}^2
\nonumber\\
+
\frac{1}{J}\sum_{j=1}^{J}\frac{1}{N_j}
\Big\|
    \sum\nolimits_{q=0}^{Q-1}
    (\alpha_{{\rm XS},q}^\star-\alpha_{j,q}^\star)
    \Hmat_{j,q}
\Big\|_{\Cyyj{j}}^2.\hspace{0.75cm}
\label{eq:XS_decomposition}
\eeqna
The first term in \eqref{eq:XS_decomposition} is the within-group
graph-filter approximation error, whereas the second is the additional
cost imposed by sharing a single spectral function across network groups.

\vspace{-0.25cm}
\subsection{Interpretation and Implementation}
\label{subsec:prop}

\begin{remark}[Spectral interpretation]
\label{rem:spectral_interpretation}
When each group contains a single known graph, the cross-network design
admits a simple spectral interpretation. Each graph has its own optimal
GSP-\ac{lmmse} response, sampled at its graph frequencies. The proposed
design fits a single function $h(\lambda;\alphavec)$ to these responses
jointly across all graphs, thereby allowing graphs of different sizes to
share the same spectral mapping.

In particular, define the graph-specific GSP-\ac{lmmse} response
\begin{equation}
    [\fstar{j}]_n
    \define
    \frac{
    [\diag(\Vmat_j^T\Cxyj{j}\Vmat_j)]_n}
    {[\diag(\Vmat_j^T\Cyyj{j}\Vmat_j)]_n},
    \qquad n=1,\ldots,N_j .
    \label{eq:crosssize_per_graph_response}
\end{equation}
Hence, Theorem~\ref{thm:crosssize} fits a common spectral function
$h(\lambda;\alphavec)$ to the graph-specific GSP-\ac{lmmse} responses
across all networks.
Then, minimizing the cross-network MSE in
\eqref{eq:crosssize_risk} is equivalent to the weighted curve-fitting problem
\beqna
\alphavec_{\rm XS}^\star
=
\argmin_{\alphavecsmall \in\mathbb R^Q }
\sum_{j=1}^{J}\frac{1}{N_j}
\sum_{n=1}^{N_j}
[\diag(\Vmat_j^T\Cyyj{j}\Vmat_j)]_n
\nonumber\\
\times
\left(
h(\lambda_{j,n};\alphavec)-[\fstar{j}]_n
\right)^2 .
\label{eq:curvefit_risk}
\eeqna
Thus, \eqref{eq:curvefit_risk} is a weighted curve-fitting problem:
the samples $(\lambda_{j,n},[\fstar{j}]_n)$ from all graphs are fitted
by the common spectral function $h(\lambda;\alphavec)$, with each
frequency weighted by its measurement power. 
If $\mathbf{1}^T\mathbf{x}=0$, as in the \ac{psse}
simulations, the DC mode is excluded from the
estimate, and the sum in \eqref{eq:curvefit_risk}
starts at $n=2$.
\end{remark}


Beyond accommodating different graph sizes, the shared spectral-function
parameterization offers two additional advantages.
First, the estimation approach can be performed without explicit \ac{evd}:
Chebyshev graph filters can be implemented without 
\ac{evd} using sparse Laplacian--vector products
~\cite{Chebyshev_2018},
while rational filters can be implemented using sparse linear-system
solvers or \ac{arma} graph filters~\cite{Isufi_Leus2017}.
Second, the learned spectral function can be applied directly to a new, 
unseen graph $\Lmat_0$ whose topology is
known at inference. The corresponding estimator is
\begin{equation}
\hat\xvec_0(\yvec_0)
=
\muxj{0}
+
h(\Lmat_0;\alphavec_{\rm XS}^\star)
(\yvec_0-\muyj{0}).
\label{eq:unseen_estimator}
\end{equation}
Thus, no retraining of the spectral coefficients is required,
provided the target graph has compatible spectral scaling and
the target means are available.

\vspace{-0.25cm}

\section{Sample-Based Estimator Design}
\label{sec:samples}
\vspace{-0.05cm}
The estimators in Sections~\ref{sec:multibasis} and \ref{sec:crosssize} minimize population MSEs, and
depend on generally unknown statistical moments.
Given labeled training data, 
obtained from historical measurements or generated using the model in \eqref{eq:nonlinear_model}, they can instead be designed by
minimizing the corresponding empirical MSEs.
 This model-based discriminative approach \cite{dabush2023generative,shlezinger2022discriminative}
preserves the graph-filter structure while optimizing
its coefficients directly from estimation
errors.

\vspace{-0.25cm}
\subsection{Fixed-Dimension Estimators}
\label{subsec:sample_mb}
Suppose that for each
topology $m$, with known $\Lmat_m$, a training set
$\mathcal D_m=
\{(\xvec_m^{(t)},\yvec_m^{(t)})\}_{t=1}^{T_m}$
of \ac{iid} realizations is available.
The number of samples $T_m$ may vary across topologies.
 
We define the
prior-weighted sample means
\begin{equation}
\widehat{\mubx}
=
\sum_{m=1}^{M}\frac{p_m}{T_m}
\sum_{t=1}^{T_m}\xvec_m^{(t)},
~~~
\widehat{\muby}
=
\sum_{m=1}^{M}\frac{p_m}{T_m}
\sum_{t=1}^{T_m}\yvec_m^{(t)}
\label{eq:sample_means}
\end{equation}
and the corresponding empirical mixture covariances 
\beqna
\widehat{\Cxx}
\define
\sum_{m=1}^{M}\frac{p_m}{T_m}
\sum_{t=1}^{T_m}
(\xvec_m^{(t)}-\widehat{\mubx})
(\xvec_m^{(t)}-\widehat{\mubx})^T,
\label{eq:59}
\\
\widehat{\Cxy}
\define
\sum_{m=1}^{M}\frac{p_m}{T_m}
\sum_{t=1}^{T_m}
(\xvec_m^{(t)}-\widehat{\mubx})
(\yvec_m^{(t)}-\widehat{\muby})^T,
\label{eq:60}
\\
\widehat{\Cyy}
\define
\sum_{m=1}^{M}\frac{p_m}{T_m}
\sum_{t=1}^{T_m}
(\yvec_m^{(t)}-\widehat{\muby})
(\yvec_m^{(t)}-\widehat{\muby})^T.
\label{eq:sample_mixture_covariances}
\eeqna

For the multi-basis estimator from Subsection \ref{subsec:dictionary}, consider the weighted empirical MSE
\begin{equation}
\widehat{\Rcal}_{\rm MB}(\etavec)
\define
\sum_{m=1}^{M}\frac{p_m}{T_m}
\sum_{t=1}^{T_m}
\left\|
\widehat\xvec_{\rm MB}
(\yvec_m^{(t)};\etavec)-\xvec_m^{(t)}
\right\|^2,
\label{eq:empirical_MB_risk}
\end{equation}
where $\widehat\xvec_{\rm MB}(\cdot;\etavec)$ has the form
of \eqref{eq:MB_estimator}, with the population
means replaced by \eqref{eq:sample_means}. It can be shown in a similar manner to \cite{kroizer2021bayesian} that
for any fixed linear operator, these empirical means yield
the MSE-minimizing affine offset.
Expanding \eqref{eq:empirical_MB_risk} gives
\begin{equation}
\widehat{\Rcal}_{\rm MB}(\etavec)
=
\tr(\widehat\Cxx)
-2\etavec^T\widehat\rvecD
+\etavec^T\widehat\SmatD\etavec,
\label{eq:sample_MB_quadratic}
\end{equation}
where $\widehat\Cxx$ is defined in
\eqref{eq:59}, and
$\widehat\SmatD$ and $\widehat\rvecD$ are obtained
from \eqref{eq:S_blocks}--\eqref{eq:r_blocks} by
replacing $\Cyy$ and $\Cxy$ with their empirical
counterparts in \eqref{eq:60}--\eqref{eq:sample_mixture_covariances}.
Consequently, the minimum-norm empirical MSE minimizer is
\begin{equation}
\widehat\etavec
=
\widehat\SmatD^\dagger\widehat\rvecD.
\label{eq:sample_eta}
\end{equation}
Thus, the sample-based multi-basis estimator coincides with the
empirical MSE minimizer within the prescribed estimator class,
rather than merely a heuristic plug-in approximation.
The same principle applies to the shared-response estimator
in \eqref{eq:shared_estimator}. For fixed topology weights,
minimizing its weighted empirical MSE yields
\eqref{eq:shared_solution}, with $\Cxy$ and $\Cyy$
replaced by their empirical counterparts.

For the topology-specific estimators in
Subsection~\ref{subsec:topology_selection}, let
$\widehat{\muxm{m}}$, $\widehat{\muym{m}}$,
$\widehat{\Cxym{m}}$, and $\widehat{\Cyym{m}}$
denote the sample means and covariances conditioned on
topology $m$. Define
\begin{equation}
\widehat{\dxym{m}}
\define
\diag(\Vmat_m^T\widehat{\Cxym{m}}\Vmat_m),
~
\widehat{\dyym{m}}
\define
\diag(\Vmat_m^T\widehat{\Cyym{m}}\Vmat_m).
\label{eq:sample_dxy_dyy}
\end{equation}
Minimizing the corresponding conditional empirical MSEs
gives the topology-specific responses
\begin{equation}
\widehat\hvec_{{\rm GSP},m}
=
\widehat{\dxym{m}}
\oslash
\widehat{\dyym{m}},
\label{eq:sample_topology_specific_response}
\end{equation}
whereas minimizing their prior-weighted average under
a common-response constraint gives
\begin{equation}
\widehat\hvec_{\rm SR}
=
\left(
\sum\nolimits_{m=1}^{M}p_m\widehat{\dxym{m}}
\right)
\oslash
\left(
\sum\nolimits_{m=1}^{M}p_m\widehat{\dyym{m}}
\right).
\label{eq:sample_topology_aware_response}
\end{equation}
These responses are substituted into
\eqref{eq:topology_specific_gsp} and
\eqref{eq:topology_aware_estimator}.
For topology-selective estimation, the learned
topology-specific filters are combined with the selection
rule of Subsection~\ref{subsec:topology_selection}.

\vspace{-0.25cm}
\subsection{Cross-Size Shared-Function Estimator}
\label{subsec:sample_xs}

For each group $j$, suppose that
$\{(\xvec_j^{(t)},\yvec_j^{(t)})\}_{t=1}^{T_j}$
are \ac{iid} samples drawn from its within-group mixture.
Let $\widehat{\muxj{j}}$, $\widehat{\muyj{j}}$,
$\widehat{\Cxyj{j}}$, and $\widehat{\Cyyj{j}}$
denote the corresponding empirical means and covariances.

For fixed topology weights and spectral dictionary,
the common coefficients are learned by minimizing the
average per-node empirical MSE
\be
\widehat{\Rcal}_{\rm XS}(\alphavec)
\define{}
\frac{1}{J}\sum_{j=1}^{J}
\frac{1}{N_jT_j}
\sum_{t=1}^{T_j}
\left\|
\widehat\xvec_j(\yvec_j^{(t)};\alphavec)
-\xvec_j^{(t)}
\right\|^2,
\label{eq:empirical_XS_risk}
\ee
where $\widehat\xvec_j(\cdot;\alphavec)$ has the form
of \eqref{eq:crosssize_estimator}, with the group means
replaced by their empirical counterparts.
As in Theorem~\ref{thm:crosssize}, the empirical objective
is quadratic in $\alphavec$. Its minimum-norm minimizer is
\begin{equation}
\widehat\alphavec_{\rm XS}
=
\widehat{\mathbf S}_{\rm XS}^{\dagger}
\widehat{\mathbf r}_{\rm XS},
\label{eq:sample_crosssize_alpha}
\end{equation}
where, for $q,r=0,\ldots,Q-1$,
\begin{align}
[\widehat{\mathbf S}_{\rm XS}]_{q+1,r+1}
&\define
\sum\nolimits_{j=1}^{J}
\frac{
\tr(\Hmat_{j,q}\widehat{\Cyyj{j}}\Hmat_{j,r}^T)
}{N_j},
\nonumber\\
[\widehat{\mathbf r}_{\rm XS}]_{q+1}
&\define
\sum\nolimits_{j=1}^{J}
\frac{
\tr(\Hmat_{j,q}\widehat{\Cxyj{j}}^T)
}{N_j}.
\label{eq:sample_XS_matrices}
\end{align}
Substituting \eqref{eq:sample_crosssize_alpha} and the
empirical group means into \eqref{eq:crosssize_estimator}
gives the sample-based cross-network estimator.

Hence, the learned spectral coefficients directly minimize
the empirical estimation error across all network groups,
while their number remains independent of the graph
dimensions. In particular, this discriminative design
does not require identifying the underlying nonlinear
measurement model.

\vspace{-0.25cm}
\subsection{Statistical Complexity}
\label{subsec:sample_complexity}

The proposed estimator classes differ
in the number of coefficients learned from the training data.
Table~\ref{tab:complexity} summarizes these dimensions and
identifies the corresponding estimator formulations.
\begin{table}[t]
\centering
\caption{Coefficient dimensions of the estimator classes.}
\label{tab:complexity}
\vspace{-0.15cm}
\scriptsize
\setlength{\tabcolsep}{2.5pt}
\begin{tabular}{@{}p{0.37\columnwidth}p{0.40\columnwidth}c@{}}
\hline
Estimator & Coefficient dimension & Cross-size\\
\hline
Topology-blind \ac{lmmse} \eqref{eq:global_lmmse}
& $N^2$
& No\\
Multi-basis \eqref{eq:MB_estimator}
& $MN$; effective $\leq 1+M(N-1)$
& No\\
Rank-$r$ multi-basis \eqref{eq:lowrank_factor}
& $r(N+M-r)$
& No\\
Shared response \eqref{eq:shared_estimator}
& $N$ (fixed $\omegavec$)
& No\\
Shared spectral function \eqref{eq:crosssize_estimator}
& $Q$
& Yes\\
\hline
\end{tabular}
\end{table}
The topology-blind \ac{lmmse} estimator in
\eqref{eq:global_lmmse} uses an unrestricted
$N\times N$ matrix, requiring $N^2$ coefficients.
The multi-basis estimator in \eqref{eq:MB_estimator}
assigns $N$ spectral coefficients to each of the
$M$ candidate graphs, yielding $MN$ coefficients.
However, since the identity matrix belongs to any
graph-filter basis, the representation is nonunique,
and the dimension of the induced operator space is
at most $1+M(N-1)$.
The rank-$r$ extension in \eqref{eq:lowrank_factor}
represents the $N\times M$ matrix of spectral responses
by a rank-$r$ factorization, with $r(N+M-r)$
degrees of freedom.
The shared-response estimator in
\eqref{eq:shared_estimator} uses a single
$N$-dimensional response for fixed  weights,
while the cross-size estimator in
\eqref{eq:crosssize_estimator} learns only $Q$
coefficients of a common spectral function,
independently of the graph dimensions.
These parameterizations reflect different tradeoffs between
model expressiveness and the estimated number of coefficients. Lower-dimensional models may reduce estimation
variance at the cost of increased MSE.


%

\newif\ifrationalbasis
\rationalbasistrue   

\vspace{-0.25cm}
\section{Simulations}
\label{sec:simulations}
We evaluate the proposed estimators for \ac{psse}, which is a fundamental problem in power system monitoring and operation \cite{Abur_Gomez_book,Giannakis_Wollenberg_2013}. We address five questions:
(Q1)~What is the price of an unknown topology, and how well do existing approaches cope with it (Subsection~\ref{ssec:price})?
(Q2)~What does each level of the estimator hierarchy contribute (Subsection~\ref{ssec:hierarchy})?
(Q3)~How much training data is needed (Subsection~\ref{ssec:training})?
(Q4)~How much topology information do the measurements carry (Subsection~\ref{ssec:selection})?
(Q5)~Can a spectral response be transferred across network sizes (Subsection~\ref{ssec:transfer})?
The code that reproduces all the figures and tables is available at \url{https://github.com/tirzar/GSP_LMMSE_Topo}.

\vspace{-0.25cm}
\subsection{Setup}
\label{ssec:setup}
\emph{Networks and models:} We use six IEEE test systems from MATPOWER~\cite{zimmerman2011matpower} (Table~\ref{tab:networks}). Each power grid is modeled as a weighted graph whose nodes are buses
(loads and generators) and whose edges represent transmission lines and transformers. The edge weights are derived from the branch series
susceptances using the bus admittance matrix $\Ymat$:
$[\Wmat]_{i,j}=\frac{1}{2}\Im\{Y_{i,j}+Y_{j,i}\}$
for $i\neq j$, with zero diagonal.
The Laplacian is
$\Lmat=\diag(\Wmat\mathbf{1})-\Wmat$.

\input{tables/tab_A_networks}

The state $\xvec$ contains the voltage phase angles, referenced to their mean ($\mathbf{1}^T\xvec=0$), and $\mathbf{y}$ contains the active-power injections at all the buses. Thus, we use the measurement model \eqref{eq:nonlinear_model} with $\mathbf{w}\sim\mathcal{N}(\mathbf{0},\sigma^2\mathbf{I}_N)$ and $\text{SNR}=-10\log_{10}\sigma^2$ [dB], where $\sigma$ is in per unit (p.u.) on a 100-MVA base. By the AC power-flow equations, the measurement function is
\begin{equation}
[\gvec(\mathbf{L}_m,\xvec)]_n=\sum_{k=1}^{N}\Big(G^{(m)}_{n,k}\cos(x_n-x_k)+B^{(m)}_{n,k}\sin(x_n-x_k)\Big),
\label{eq:ac_pf}
\end{equation}
$n=1,\ldots,N$, where $Y^{(m)}_{n,k}=G^{(m)}_{n,k}+jB^{(m)}_{n,k}$ is the $(n,k)$th entry of the bus admittance matrix of topology $m$, and all the voltage magnitudes are set to $1$\,p.u.

We consider two state-generation models:
\begin{itemize}
    \item Model M1 (smoothness prior): $\xvec|m^\star{=}m\sim\mathcal{N}(\mathbf{0},\beta\mathbf{L}_m^\dagger)$ with $\beta=0.1$, the graph-smoothness prior of~\cite{kroizer2021bayesian}, which is defined through the same Laplacian as the measurements.
    \item Model M2 (power-flow states): $\xvec$ solves $[\gvec(\mathbf{L}_m,\xvec)]_n=p_n$ at all the non-slack buses $n$ for random specified injections $p_n$. The loads are scaled by a common factor drawn from $U[0.8,1.2]$ and perturbed independently by Gaussian deviations with a 10\% standard deviation per bus; the generators share the total load in proportion to their base dispatch; and the slack bus balances the losses.
\end{itemize}
Both models yield comparable phase-angle magnitudes
(Table~\ref{tab:networks}).
We use nominal \acp{snr} of $20$ and $40$\,dB for M1
and M2, respectively ($\sigma=10$ and $1$\,MW),
yielding comparable injection-to-noise ratios.

\emph{Candidate topologies:}
For most simulations, we consider the IEEE 118-bus and NE 39-bus systems~\cite{zimmerman2011matpower}.
For each system, the candidate family comprises the
nominal topology and $19$ variants ($M=20$), with
uniform prior probabilities. Each variant is generated
by independently removing existing lines with probability
$0.01$ and adding $A\sim\mathrm{Poisson}(1)$ lines
between nonadjacent buses sharing a neighbor in the
nominal network.
The variants differ from the nominal topology by
1--9 lines (IEEE 118) and 1--6 lines (NE 39).
Despite these limited topology changes, the
index-wise \ac{gft} eigenvector alignment
$\frac{1}{N}\sum_{n=1}^N
\bigl([\Vmat_1]_{:,n}^T[\Vmat_m]_{:,n}\bigr)^2$
can be as low as $0.13$ and $0.21$, respectively, i.e., the candidates have
substantially different graph-frequency bases.
Unless stated otherwise, $T_m=500$ training samples
are used per topology. In Subsections~\ref{ssec:price}--\ref{ssec:selection},
the results are averaged over $20$ independent realizations of the candidate
families, training sets, and test sets. The error bars in the figures are
95\% confidence intervals, computed as $\pm1.96$ standard errors of the
per-realization NMSE in dB.

\input{tables/tab_A_methods}

\emph{Compared estimators (Table~\ref{tab:methods}):}
An estimator is called topology blind if the active topology is not revealed
to it; the data-dependent estimators infer it from $\yvec$.
The robust minimax estimators~\cite{Eldar2005robust,Eldar2004competitive}
use the candidate topologies as the uncertainty set.
We compute their least-favorable mixture weights by
mirror ascent and use the corresponding LMMSE estimator.
The classical multiple-model estimator~\cite{Magill1965,LiJilkov2005} uses
unrestricted LMMSE branches with full-covariance Gaussian likelihoods,
while the proposed variant uses GSP-LMMSE branches with the
graph-stationary surrogate likelihood
$p(\mathbf{y}|m^\star=m)=\mathcal{N}\big(\boldsymbol{\mu}_y^{(m)},
\mathbf{V}_m\mathrm{diag}(\mathbf{d}_{yy}^{(m)})\mathbf{V}_m^T\big)$,
which requires only $N$ variances per candidate and is diagonal in the
GFT domain.
The DNN~\cite{Mestav2019} is trained on pooled data
by minimizing the MSE.
In Subsection~\ref{ssec:hierarchy}, we also include, as an ablation, the
spectral response learned from the pooled data in the nominal basis only
(``single basis $\Vmat_1$'', i.e., Proposition~\ref{prop:shared_blind} with
$\omega_m=\delta[m-1]$).
Sample covariance inversions employ cross-validated
diagonal loading.
Further implementation details regarding the different estimators are provided in the
supplementary material.
Performance is evaluated using the NMSE,
$R(\hat{\xvec})/\mathbb{E}[\|\xvec-\boldsymbol\mu_x\|^2]$, and
the worst-topology conditional NMSE.

\subsection{The Price of an Unknown Topology}
\label{ssec:price}

\begin{figure}[t]
\centering
\includegraphics[width=\columnwidth]{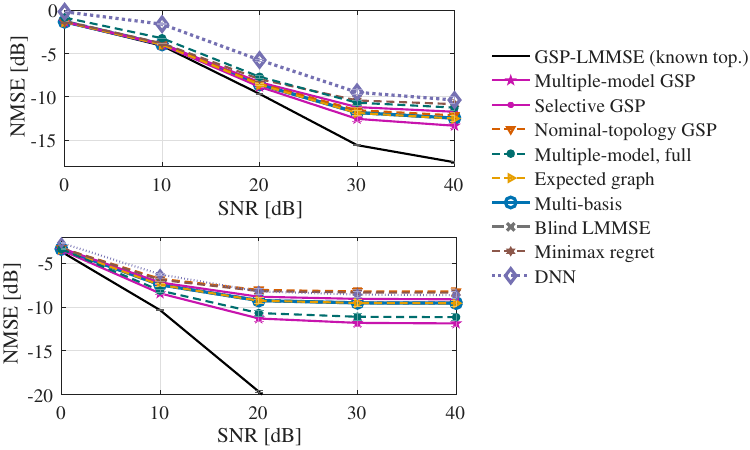}
\vspace{-0.5cm}
\caption{NMSE versus \ac{snr} under topology uncertainty
(M1, $M=20$, 20 realizations): (Top) IEEE 118-bus and
(Bottom) NE 39-bus. Error bars indicate 95\% confidence intervals. On NE 39, the topology-aware benchmark falls below the plotted range above 20 dB.}
\label{fig:price}
\end{figure}

Fig.~\ref{fig:price} shows that the topology-aware GSP-LMMSE
estimator improves with \ac{snr}, whereas the
topology-blind linear estimators exhibit an error floor.
The gap between the topology-aware benchmark and the best
topology-blind linear estimator increases from less than
$0.3$\,dB at $0$\,dB \ac{snr} to $5.0$\,dB (IEEE 118)
and $25.5$\,dB (NE 39) at $40$\,dB.
This is consistent with~\eqref{eq:linear_price},
which characterizes the loss when a single affine
operator serves all candidate topologies.
Multiple-model GSP reduces this loss relative to blind LMMSE
by $0.3$--$0.9$\,dB (IEEE 118) and $2.1$--$2.3$\,dB
(NE 39) at $20$--$40$\,dB.
It can be seen that selective GSP remains below the topology-aware benchmark even at high \ac{snr}. This is because the topology-conditioned measurement distributions
are difficult to distinguish from a single observation,
and the identification rate saturates above $20$\,dB.
Consequently, soft weighting outperforms hard selection
by $1.6$ and $2.8$\,dB at $40$\,dB.

\emph{Linear estimators:}
Table~\ref{tab:comparison} compares the estimators at
$20$\,dB \ac{snr}.
The blind LMMSE, multi-basis, rank-one, and expected-graph
estimators achieve NMSEs within $0.1$\,dB under M1.
The structured estimators require only $N$--$MN$
coefficients, versus $N^2$ for blind LMMSE.
Relative to multi-basis, nominal-topology
GSP~\cite{kroizer2021bayesian} loses $0.1$ and
$1.2$\,dB on average and $0.8$ and $3.2$\,dB
for the worst topology on IEEE 118 and NE 39, respectively.
The minimax estimators
\cite{Eldar2005robust,Eldar2004competitive}
sacrifice $0.6$--$1.1$\,dB in average NMSE to improve
worst-topology NMSE by $0.2$ and $1.2$--$1.3$\,dB.
On IEEE 118, even topology-aware unrestricted LMMSE
underperforms blind LMMSE: it estimates
$N^2=13{,}924$ coefficients from $T_m=500$ samples
per topology, whereas blind LMMSE pools
$MT_m=10^4$ samples.

\input{tables/tab_B_comparison}

\emph{Data-dependent estimators:}
Multiple-model GSP achieves the best NMSE performance among topology-blind estimators.
Relative to blind LMMSE, it improves average NMSE by $0.3$ and $2.1$\,dB across all 20 realizations, and worst-topology NMSE by $0.6$ and $2.8$\,dB. These gains do not require accurate topology identification:
the rates are only $21\%$ and $26\%$, versus
$5\%$ for random selection. Posterior weighting
suppresses unlikely candidates without identifying
the active topology.
Hard selection using the same posterior performs
$0.3$ and $0.4$\,dB worse than multi-basis on average, but improves worst-topology NMSE by $0.5$ and $2.0$\,dB.
The classical full-covariance multiple-model estimator
\cite{Magill1965,LiJilkov2005} performs $1.2$ and
$0.6$\,dB worse than multiple-model GSP, with lower identification rates ($9\%$ and $19\%$), and the DNN~\cite{Mestav2019} is $3.1$\,dB worse on both networks, with a design time of $66$\,s vs. $4$\,ms.

\vspace{-0.25cm}
\subsection{Estimator Hierarchy: Number of Bases and
Perturbation}
\label{ssec:hierarchy}

\begin{figure}[t]
\centering
\includegraphics[width=\columnwidth]{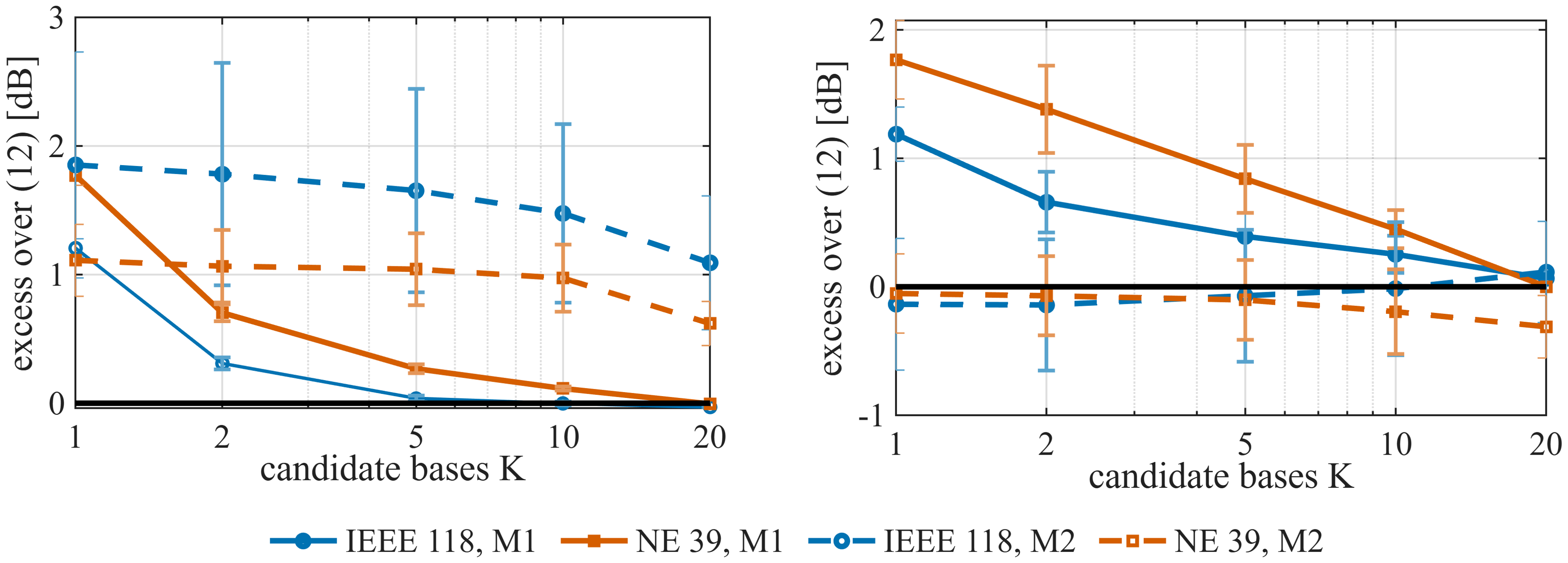}
\caption{Excess NMSE of the multi-basis estimator using
$K$ of the $M=20$ candidates relative to the
topology-blind LMMSE estimator~\eqref{eq:global_lmmse}:
(Left) average over topologies; (Right) worst topology.}
\label{fig:bases}
\end{figure}

\begin{figure}[t]
\centering
\includegraphics[width=\columnwidth]{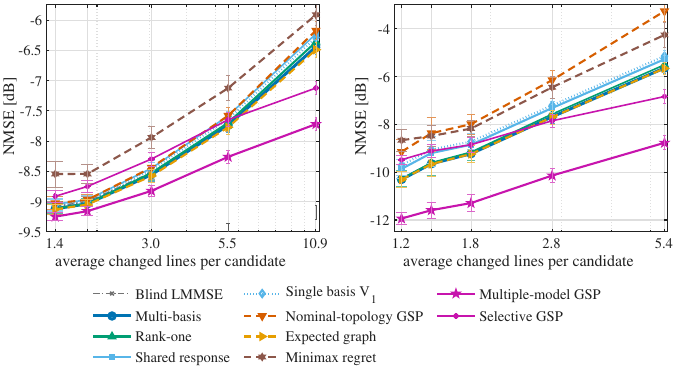}
\caption{NMSE versus the average number of changed lines
per candidate topology (M1, $20$\,dB):
(Left) IEEE 118-bus  and (Right) NE 39-bus.
The omitted topology-aware GSP-LMMSE benchmark attains
$-9.7$ to $-9.2$\,dB and $-19.7$ to $-18.9$\,dB.}
\label{fig:perturbation}
\end{figure}

\emph{Restricted estimator classes:}
The rank-one estimator incurs at most $0.06$\,dB
additional NMSE relative to multi-basis, while the
shared-response estimator incurs $0.1$\,dB (IEEE 118)
and $0.4$\,dB (NE 39) on average, reaching $1.3$\,dB
for the worst topology on NE 39 (Table~\ref{tab:comparison}).
These results are consistent with the nested estimator
classes in~\eqref{eq:single_basis_dominance}
and~\eqref{eq:shared_response_gap}.
Thus, rank-one factorization retains nearly the full
multi-basis performance with only $N+M-1$ coefficients.

\emph{Number of bases:}
Fig.~\ref{fig:bases} examines the effect of restricting
the multi-basis estimator to $K\le M$ candidate bases.
Under M1, increasing $K$ from $1$ to $5$ reduces
the average excess NMSE from $1.2$ to $0.04$\,dB
(IEEE 118) and from $1.8$ to $0.27$\,dB (NE 39).
With all $M=20$ bases, the excess is negligible,
consistent with the projection characterization
in Theorem~\ref{thm:projection}.
The worst-topology excess decreases more slowly.
Under M2, an average excess of $1.1$ and $0.6$\,dB
remains even with all bases, reflecting the restriction
to the multi-basis operator space.
Nevertheless, multi-basis improves upon single-basis
and expected-graph filtering by $0.5$--$0.8$\,dB
on average. Its worst-topology excess is at most $0.12$\,dB and can be negative, since
\eqref{eq:global_lmmse} minimizes the average  MSE.

\emph{Perturbation size:}
Fig.~\ref{fig:perturbation} examines increasing topology
uncertainty by scaling the line-removal probability
and expected number of added lines by $0.25$--$4$.
The resulting candidates differ from the nominal graph
by $1.4$--$10.9$ lines (IEEE 118) and $1.2$--$5.4$
lines (NE 39) on average.
While the topology-aware benchmark varies by less
than $0.8$\,dB, the topology-blind linear estimators
deteriorate substantially. On NE 39, for example,
the multi-basis NMSE increases from $-10.3$ to
$-5.6$\,dB, compared with $-9.1$ to $-3.3$\,dB
for nominal-topology GSP~\cite{kroizer2021bayesian}.
The single-basis estimator, whose response is learned
from pooled data in the nominal basis $\mathbf V_1$,
helps isolate the gains from each level of the hierarchy.
At the baseline perturbation level on NE 39,
learning this response improves upon nominal-topology
GSP by $0.7$\,dB; combining all candidate bases
provides another $0.5$\,dB, and multiple-model GSP
adds $2.1$\,dB through data-dependent weighting.
Moreover, the gain of multiple-model GSP over
multi-basis increases with perturbation size,
from $0.1$ to $1.3$\,dB (IEEE 118) and from
$1.6$ to $3.1$\,dB  (NE 39).
Hard selection also benefits from larger perturbations,
but remains $0.3$--$0.6$ and $1.9$--$2.5$\,dB
worse than soft weighting on the respective networks.

\vspace{-0.25cm}
\subsection{Finite Training Data}
\label{ssec:training}

Fig.~\ref{fig:training} shows the NMSE versus the number
of training samples per topology, $T_m$, for IEEE 118
under M1 and M2. Similar trends are observed on NE 39.

\emph{Structure as regularization (M1):}
The estimators with $O(N)$ coefficients are nearly
insensitive to $T_m$ (Fig.~\ref{fig:training}(left)).
At $T_m=20$, blind LMMSE ($N^2=13{,}924$ coefficients)
is $1.3$\,dB worse than multi-basis and $1.5$\,dB
worse than rank-one and expected-graph filtering,
catching up only for $T_m\ge1,000$.
Topology-aware unrestricted LMMSE also underperforms
the blind structured estimators up to $T_m=500$.
The graph-stationary surrogate improves data efficiency:
the classical multiple-model estimator requires
$T_m=2,000$ to match the NMSE achieved by multiple-model
GSP with $T_m=50$.
Multiple-model GSP performs best among blind estimators
for $T_m\ge100$, whereas rank-one and expected-graph
filtering are preferable with fewer samples.
Selective GSP and the DNN remain $0.3$ and $1.0$\,dB
worse than multi-basis, respectively, even with
$4\times10^4$ pooled samples for the DNN.

\emph{Beyond graph stationarity (M2):}
The price of unknown topology is substantial:
the topology-aware GSP-LMMSE and LMMSE estimators,
omitted from Fig.~\ref{fig:training}(right) for clarity,
attain $-18.8$ and $-20.3$\,dB at $T_m=2,000$,
compared with at best $-5.3$\,dB without topology
knowledge.
Without graph stationarity, topology-aware GSP-LMMSE
trails unrestricted LMMSE by $1.5$\,dB.
Blind LMMSE outperforms multi-basis by $0.7$--$1.1$\,dB,
while multi-basis improves upon single-graph designs
by $0.6$--$0.7$\,dB.
Under M2, topology affects the measurement distribution
primarily through the losses, making the graph-stationary
surrogate likelihood inaccurate. Consequently,
multiple-model and selective GSP perform worst
($-0.9$ to $0.5$\,dB).
With $T_m=2,000$, classical multiple-model estimation,
using unrestricted branches and full-covariance
likelihoods, achieves the lowest NMSE.
Thus, the data-dependent estimators of
Subsection~\ref{subsec:topology_selection} are most effective when the topology is reflected in the graph-spectral covariance of the measurements, as in M1.


\begin{figure}[t]
\centering
\includegraphics[width=\columnwidth]{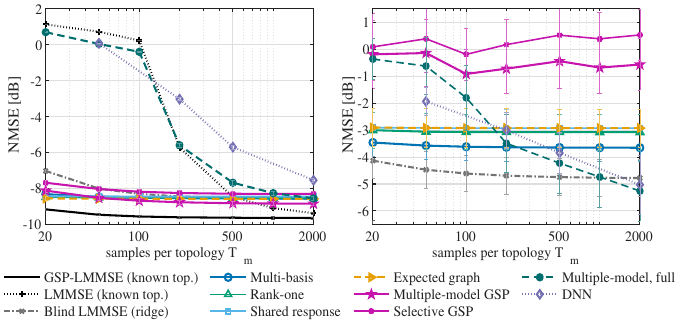}
\caption{NMSE versus the training samples per topology $T_m$ on the IEEE
118-bus system (20 realizations): M1, $20$\,dB (left), and M2, $40$\,dB (right).}
\label{fig:training}
\end{figure}

\vspace{-0.25cm}
\subsection{Topology Information in the Measurements}
\label{ssec:selection}
We now consider $P$ conditionally independent snapshots
$\yvec^{(1)},\ldots,\yvec^{(P)}$ sharing the same active topology.
The likelihood in~\eqref{eq:map_topology_selector} is then replaced by
$\prod_{i=1}^{P}p(\yvec^{(i)}|m^\star=m)$.
The resulting joint likelihood is also used to compute the
posterior weights of the multiple-model estimators, with each
snapshot estimated using the selected or posterior-weighted branches.

Fig.~\ref{fig:selection} presents the results for IEEE 118.
With $T_m=500$, the surrogate-based identification rate increases
from $0.21$ ($P=1$) to $0.68$ ($P=10$) and $0.87$ ($P=50$),
whereas the full-covariance likelihood reaches only $0.39$.
With $T_m=50<N$, reliable estimation of the full covariance
matrices is not possible, and the classical rule remains near chance
($0.08$ at $P=50$). In contrast, the graph-stationary surrogate,
which requires only $N$ variances per candidate, reaches $0.59$
(Fig.~\ref{fig:selection}, left).
Correspondingly, multiple-model GSP improves from $-8.8$ to
$-9.6$\,dB, within $0.03$\,dB of the topology-aware benchmark
(Fig.~\ref{fig:selection}, right).
Soft weighting provides the largest gain for small $P$
($0.5$\,dB over hard selection at $P=1$), when the candidate
topologies are difficult to distinguish. At $P=50$, hard
selection achieves nearly identical NMSE (within $0.01$\,dB)
while also providing a topology decision.
Thus, identification accuracy alone does not determine
estimation performance.
Similar trends are observed on NE 39, with larger gains:
the identification rate increases from $0.26$ to $0.92$,
and multiple-model GSP improves from $-11.3$ to $-19.3$\,dB,
compared with $-19.7$\,dB for the topology-aware benchmark.

\begin{figure}[t]
\centering
\includegraphics[width=\columnwidth]{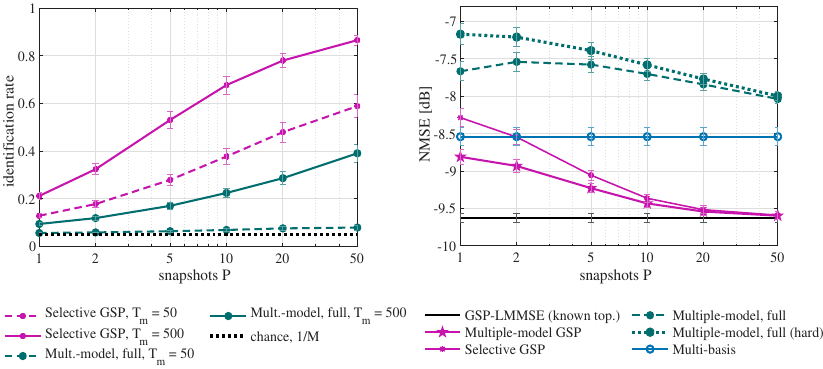}
\vspace{-0.25cm}
\caption{Topology selection using $P$ snapshots  (IEEE 118, M1, $20$\,dB, 20 realizations):
(Left) identification rate of the approximate \ac{map}
rule~\eqref{eq:map_topology_selector} via the graph-stationary
surrogate (selective GSP) and  full-covariance likelihood
\cite{Magill1965}, for  $T_m=50$ (dashed) and $500$ (solid); 
(Right) NMSE with soft weighting and hard selection ($T_m=500$).}
\label{fig:selection}
\end{figure}

\subsection{Transfer Across Network Sizes}
\label{ssec:transfer}

\input{tables/tab_F_loo}

We evaluate the cross-size design of Section~\ref{sec:crosssize} on the six
networks with known topologies (Remark~\ref{rem:spectral_interpretation});
the implementation details are given in the supplementary material. We
compare a Chebyshev dictionary of order $Q=60$ with a stability penalty, on
three frequency axes: absolute, normalized $\lambda/\lambda_{\max}$ as in
ChebNet~\cite{Defferrard2016}, and $\lambda/N$ as in graphon signal
processing~\cite{Ruiz_Chamon_Ribeiro_GSP_2021}; and a rational dictionary
$\phi_q(\lambda)=1/(\lambda+c_q)$ with four fixed poles.
 Under M1, the responses of all networks lie on
a single curve of the \emph{absolute} graph frequency
(Fig.~\ref{fig:collapse}, left). This is expected: since
$\gvec(\mathbf{L},\xvec)\hspace{-0.05cm}\approx\hspace{-0.05cm}\gvec(\mathbf{L},\mathbf{0})+\mathbf{L}\xvec$
and $\mathrm{Var}(\tilde x_n)=\beta/\lambda_n$, the per-frequency gain is
$\beta\hspace{-0.05cm}/\hspace{-0.05cm}(\beta\lambda_n+\sigma^2)$. Under M2, only the
high-frequency part $f(\lambda)\hspace{-0.05cm}\approx\hspace{-0.05cm}1/\lambda$, i.e., the power-flow
inversion, is common (Fig.~\ref{fig:collapse}, right).

\begin{figure}[t]
\centering
\includegraphics[width=0.95\columnwidth]{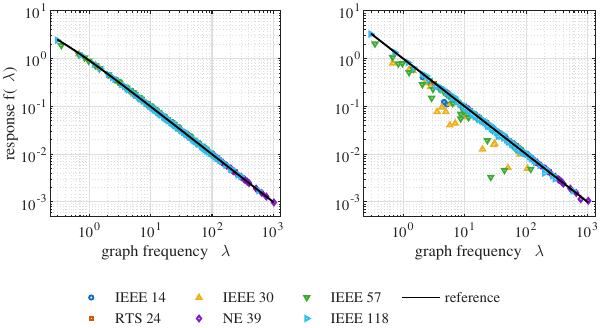}
\vspace{-0.1cm}
\caption{GSP-LMMSE responses~\eqref{eq:crosssize_per_graph_response} of the
six networks: (Left) M1, with the reference $\beta/(\beta\lambda+\sigma^2)$;
(Right) M2, with the reference $1/\lambda$.}
\label{fig:collapse}
\end{figure}

\emph{Transfer to an unseen network (Table~\ref{tab:loo}):} The function is
designed from five networks by Theorem~\ref{thm:crosssize} and applied to
the sixth by~\eqref{eq:unseen_estimator}. Under M1, the Chebyshev function
on the absolute axis transfers with an excess of $0.06$--$1.4$\,dB, except
to NE 39, whose spectrum ($\lambda_{\max}=1{,}025$) exceeds the training
spectra, so that extrapolation costs $12.4$\,dB. The normalized axes used in
the GNN and graphon literature fail for all the targets
($2.1$--$32.9$\,dB), since the graph frequencies of physical networks carry
units. The rational dictionary transfers to all six networks within
$0.12$\,dB: unlike the Chebyshev function, which requires a high order
($Q=60$ at $20$\,dB) and the stability penalty to extrapolate, it matches
the pole of $\beta/(\beta\lambda+\sigma^2)$ with four coefficients. Under
M2, the transferred rational function loses $1.5$--$7.3$\,dB, in line with
Fig.~\ref{fig:collapse} (right).

\emph{Few-shot adaptation:} When the target network has only $T_0$ labeled
samples, transfer is most valuable under M1: with $T_0=2$ and $5$, the
transferred function has an excess of $1.9$ and $0.8$\,dB (averaged over the
six networks), versus $38$ and $2.5$\,dB for a GSP-LMMSE estimator fitted on
the $T_0$ samples and $6.0$ and $1.6$\,dB for a rational fit on the target
alone; for $T_0\ge20$, all the designs are within $0.5$\,dB. Under M2, the
joint design of Subsection~\ref{subsec:sample_xs} reduces the excess of the
transferred coefficients from $3.7$--$5.4$\,dB to $1.9$--$4.7$\,dB, whereas a
rational fit on the target alone reaches $0.9$\,dB with only $T_0=10$. Thus,
the spectral shape transfers even when the coefficients do not.

\vspace{-0.25cm}
\section{Conclusions}
\label{conc_subsec}
\vspace{-0.05cm}

In this paper, we study graph-signal estimation under topology uncertainty and transfer across networks of different sizes. We derive the MSE-optimal multi-basis affine estimator and characterize its excess MSE relative to topology-blind LMMSE. We further develop lower-complexity structured estimators, data-dependent topology selection and weighting methods, and a shared spectral-function design for cross-network transfer. We evaluate these estimators for PSSE in electrical networks.
The simulations show that when topology is reflected in the graph-spectral measurement covariance, posterior weighting of GSP-LMMSE branches achieves the lowest average and worst-topology MSEs among topology-blind estimators. It does not require accurate topology identification and approaches the topology-aware benchmark as more snapshots become available. Multi-basis and rank-one estimators achieve nearly the performance of topology-blind LMMSE with substantially fewer coefficients, particularly with limited training data. When topology affects measurements mainly through the nonlinear observation model, the graph-stationary likelihood becomes inaccurate, favoring topology-blind LMMSE or, with sufficient training data, classical multiple-model estimation. Finally, shared spectral functions, particularly rational functions of absolute graph frequency, enable accurate transfer across network sizes when the optimal spectral responses exhibit common structure, while more heterogeneous responses benefit from target-network adaptation.


\vspace{-0.2cm}
\section*{Acknowledgment}
\vspace{-0.05cm}
The author thanks E. Zeltzer for his collaboration
on the preliminary work~\cite{Zeltzer_Routtenberg_CISS}
and for valuable discussions on the topic.
\appendices

\makeatletter
\def\section{\@startsection{section}{1}{\z@}%
    {1.5ex plus 1.5ex minus .5ex}%
    {0.7ex plus .5ex minus .2ex}%
    {\normalfont\normalsize\centering\scshape}}
\makeatother

\section{Proof of Theorem~\ref{thm:MB}}
\label{AppA}
Substituting \eqref{eq:MB_estimator} into \eqref{eq:bayes_risk}
yields
\begin{equation}
    \Rcal(\hat\xvec_{\rm MB})
    =
    \tr(\Cxx)
    -2\etavec^T\rvecD
    +\etavec^T\SmatD\etavec,
    \label{eq:MB_risk}
\end{equation}
where
$\etavec=[\hvec_1^T,\ldots,\hvec_M^T]^T$.
Since
\beqna
\etavec^T\SmatD\etavec
=
\tr\!\left(
\sum\nolimits_{m=1}^{M}
\Vmat_m\diag(\hvec_m)\Vmat_m^T\Cyy\right.
\hspace{1cm}\nonumber\\\times 
\left.
(\sum\nolimits_{k=1}^{M}
\Vmat_k\diag(\hvec_k)\Vmat_k^T)^T
\right)
\geq 0,
\eeqna
we have $\SmatD\succeq\zerovec$. Furthermore, boundedness of \eqref{eq:MB_risk}
requires $\rvecD\perp\Ncal(\SmatD)$;
otherwise, the objective is unbounded below
along a nullspace direction.
Thus, $\rvecD\in\range(\SmatD)$,
and the minimum-norm solution is
\eqref{eq:eta_solution}.
Substitution into \eqref{eq:MB_risk}
gives \eqref{eq:MB_minrisk}.
Finally, $\Cyy\succ\zerovec$ implies strict convexity of the MSE in
the induced operator $\sum_{m=1}^{M}
\Vmat_m\diag(\hvec_m)\Vmat_m^T$; hence, the optimal operator,
and thus the resulting estimator, is unique.

\vspace{-0.1cm}
\section{Proof of Corollaries~\ref{cor:linear_opt} and~\ref{cor:heterogeneity}}
\subsubsection{Proof of Corollary~\ref{cor:linear_opt}}
\label{App_b}
By Condition 1 and \eqref{eq:mixture_moments}, the between-topology
mean terms in the mixture covariance expressions vanish, yielding
\begin{equation}
    \Cxy=\sum_{m=1}^{M}p_m\Cxym{m},
    \qquad
    \Cyy=\sum_{m=1}^{M}p_m\Cyym{m}.
    \label{eq:common_mean_mixture_cov}
\end{equation}
By Conditions 2 and 3, both $\Cxy$ and $\Cyy$ are diagonalized by
the common eigenbasis $\Vmat$. Since $\Cyy\succ\zerovec$, its inverse
is diagonalized by the same basis, and hence
\beqna
    \Astar
    =\Cxy\Cyy^{-1} \hspace{6cm}\nonumber\\
    =\Vmat
    \left(\sum_{m=1}^{M}p_m\Vmat^T\Cxym{m}\Vmat\right)
    \left(\sum_{m=1}^{M}p_m\Vmat^T\Cyym{m}\Vmat\right)^{-1}
    \Vmat^T .\nonumber
    \label{eq:common_basis_lmmse}
\eeqna
The two matrices inside the parentheses are diagonal by Condition~3;
hence, $\Astar$ is diagonalized by $\Vmat$ and therefore belongs to
the multi-basis operator space in \eqref{iff_A_star}. The result
follows from Theorem~\ref{thm:projection}.

\subsubsection{Proof of Corollary~\ref{cor:heterogeneity}}
\label{App_B2}
Since the conditional means are topology independent,
\eqref{eq:mixture_moments} yields \eqref{eq:common_mean_mixture_cov}.
For a realization $m^\star=m$, the oracle operator in
\eqref{eq:conditional_lmmse} satisfies
$\Amat_{\rm O}^{(m)}
    =
    \Cxym{m}(\Cyym{m})^{-1}$,
and therefore
$\Cxym{m}
    =
    \Amat_{\rm O}^{(m)}\Cyym{m}$.
Substituting this result into
\eqref{eq:global_lmmse} gives
\begin{equation}
    \Astar
    =
    \left(
    \sum\nolimits_{m=1}^{M}p_m
    \Amat_{\rm O}^{(m)}\Cyym{m}
    \right)\Cyy^{-1}.
    \label{eq:Astar_oracle_form}
\end{equation}

Next, define the prior-weighted average of the oracle operators as
$\bar{\Amat}_{\rm O}
    \define
    \sum_{m=1}^{M}p_m\Amat_{\rm O}^{(m)}$.
By \eqref{eq:topology_lmmse_graph_filter},
\begin{equation}
    \bar{\Amat}_{\rm O}
    =
    \sum\nolimits_{m=1}^{M}
    \Vmat_m
    \diag(p_m\hvec_m^{\rm O})
    \Vmat_m^T,
\end{equation}
so $\bar{\Amat}_{\rm O}$ is an admissible and feasible multi-basis operator. Since
the optimal multi-basis operator is the $\Cyy$-weighted projection
of $\Astar$ onto this space, Theorem~\ref{thm:projection} implies 
\begin{equation}
    \Rcal(\hat\xvec_{\rm MB}^\star)
    -\Rcal(\hat\xvec_{\rm L}^\star)
    \leq
    \|\bar{\Amat}_{\rm O}-\Astar\|_{\Cyy}^{2}.
    \label{eq:heterogeneity_intermediate}
\end{equation}

Finally, using \eqref{eq:Astar_oracle_form} and
$\sum_m p_m\Cyym{m}=\Cyy$,
\beqna
\Astar-\bar{\Amat}_{\rm O}
\hspace{-0.05cm}=\hspace{-0.05cm}
\left(\hspace{-0.05cm}
\sum\nolimits_{m=1}^{M}p_m
\Amat_{\rm O}^{(m)}
\big(\Cyym{m}-\Cyy\big)
\right)\hspace{-0.05cm}\Cyy^{-1}.
\label{eq:heterogeneity_operator_gap}
\eeqna
Hence,
$\|\Astar\hspace{-0.05cm}-\hspace{-0.05cm}\bar{\Amat}_{\rm O}\|_{\Cyy}^{2}
\hspace{-0.15cm}=\hspace{-0.05cm}
\Big\|
\sum_{m=1}^{M}p_m
\Amat_{\rm O}^{(m)}
\big(\Cyym{m}-\Cyy\big)
\Cyy^{-1/2}
\Big\|_{\rm F}^{2}$.
Combining this identity with
\eqref{eq:heterogeneity_intermediate} proves
\eqref{eq:heterogeneity_bound}.


\clearpage
\makeatletter
\let\section\@IEEEappendixsavesection   
\gdef\Hy@chapapp{section}
\makeatother
\setcounter{section}{0}
\setcounter{equation}{0}
\setcounter{table}{0}
\setcounter{figure}{0}
\setcounter{algorithm}{0}
\renewcommand{\thesection}{S-\Roman{section}}
\renewcommand{\thesectiondis}{S-\Roman{section}}
\renewcommand{\theequation}{S\arabic{equation}}
\renewcommand{\thetable}{S-\Roman{table}}
\renewcommand{\thefigure}{S\arabic{figure}}
\renewcommand{\thealgorithm}{S\arabic{algorithm}}
\renewcommand{\theHsection}{Supp.\arabic{section}}
\renewcommand{\theHequation}{Supp.\arabic{equation}}
\renewcommand{\theHtable}{Supp.\arabic{table}}
\renewcommand{\theHfigure}{Supp.\arabic{figure}}
\providecommand{\theHalgorithm}{}\renewcommand{\theHalgorithm}{Supp.\arabic{algorithm}}

\twocolumn[{%
\begin{center}
{\Huge Supplementary Material for ``Graph Signal Estimation With Unknown Topology''}\\[0.4cm]
{\large Tirza Routtenberg, \IEEEmembership{Senior Member, IEEE}}
\end{center}
\vspace{0.3cm}}]

This supplement specifies the estimators compared in Section~\ref{sec:simulations}, as implemented in the accompanying code. Equations (S1), (S2), \dots{} are defined here; all other equation, section, and table numbers refer to the paper. 

\vspace{-0.25cm}
\section{Training Data and Common Quantities}
\label{sec:notation}
All estimators are designed from the same labeled training sets $\mathcal D_m$, $m=1,\ldots,M$, of Subsection~\ref{subsec:sample_mb}. The only exception is the nominal-topology estimator of Section~\ref{sec:competing}, which uses $\mathcal D_1$ only. The sample moments conditioned on topology $m$, i.e., $\spmx{m}$, $\spmy{m}$, $\spCxy{m}$, $\spCyy{m}$, and $\spdxy{m}$, $\spdyy{m}$ in
~\eqref{eq:sample_dxy_dyy}, are normalized by $1/T_m$. The empirical mixture moments are given by~\eqref{eq:sample_means}--\eqref{eq:sample_mixture_covariances}. For weights $\nubf$ on the probability simplex $\Delta_M$, we denote by $\widehat{\mubx}(\nubf)$, $\widehat{\muby}(\nubf)$, $\widehat{\Cxy}(\nubf)$, and $\widehat{\Cyy}(\nubf)$ the same quantities with $p_m$ replaced by $\nu_m$, so that $\widehat{\Cxy}=\widehat{\Cxy}(\mathbf p)$. The empirical \ac{mse} of a general affine estimator $\hat\xvec=\Amat\yvec+\bvec$ under topology $m$ is
\begin{eqnarray}
\spMSE_m(\Amat,\bvec)=\tr(\spCxx{m})-2\tr(\Amat\spCxy{m}^T)\nonumber\\
+\tr(\Amat\spCyy{m}\Amat^T)
+\|\spmx{m}-\Amat\spmy{m}-\bvec\|^2,
\label{eq:mse_m}
\end{eqnarray}
and its prior-weighted average is the empirical counterpart of~\eqref{eq:bayes_risk}.

\emph{Regularization.} The inverted $N\times N$ sample covariances use the diagonal loading
\begin{equation}
\Amat=\widehat{\mathbf C}_{xy}\big(\widehat{\mathbf C}_{yy}+\gamma\Imat\big)^{-1},\quad
\gamma=g\,\frac{\tr(\widehat{\mathbf C}_{yy})}{N},
\label{eq:ridge}
\end{equation}
where $g\in\{10^{-8},10^{-7.5},\ldots,10^{0}\}$ is selected by $5$-fold cross-validation of the estimation \ac{mse} on the training data of the corresponding design. The pseudo-inverses of the proposed estimators are computed with a negligible loading $10^{-10}\tr(\cdot)/\dim(\cdot)$.

\vspace{-0.25cm}
\section{Topology-Aware Benchmarks}
\label{sec:aware}
The active topology $m^\star$ is revealed to these estimators.
\emph{1. GSP-LMMSE, known topology}~\cite{S:kroizer2021bayesian}: \eqref{eq:topology_specific_gsp} with $m=m^\star$, the sample means, and the response~\eqref{eq:sample_topology_specific_response}. It uses $N$ coefficients per topology.
\emph{2. LMMSE, known topology}: the oracle estimator~\eqref{eq:conditional_lmmse} with the sample moments of topology $m^\star$, where the inverse is replaced by the loading~\eqref{eq:ridge} applied to $\mathcal D_{m^\star}$. It uses $N^2$ coefficients per topology.

\section{Proposed Estimators}
\label{sec:proposed}

These estimators use the empirical mixture moments, and $\widehat{\SmatD}$ and $\widehat{\rvecD}$ of Subsection~\ref{subsec:sample_mb}. We denote the $(m,\ell)$th $N\times N$ block of $\widehat{\SmatD}$ by $\Kmat_{m\ell}$, and the $m$th $N\times1$ block of $\widehat{\rvecD}$ by $\mathbf r_m$.

\emph{Multi-basis} (Theorem~\ref{thm:MB}): \eqref{eq:opt_est} with the coefficients~\eqref{eq:sample_eta}, i.e., $MN$ coefficients. In Subsection~\ref{ssec:hierarchy}, the dictionary of $K<M$ bases uses the corresponding blocks of $\widehat{\SmatD}$ and $\widehat{\rvecD}$, and the results are averaged over random subsets of the bases.

\emph{Shared response} (Proposition~\ref{prop:shared_blind}): \eqref{eq:shared_estimator} with $\omegavec=\mathbf p$ and the response~\eqref{eq:shared_solution}, computed from the empirical moments, i.e., $N$ coefficients. The \emph{single-basis} estimator of Subsection~\ref{ssec:hierarchy} is the same design with $\omegavec=\mathbf e_{m_0}$.

\emph{Rank-one} (Remark~\ref{rem:lowrank} with $r=1$): $\Hmat=\hvec\bvec^T$ in~\eqref{eq:lowrank_factor}, i.e., $\hvec_m=b_m\hvec$, so that $\etavec=\bvec\otimes\hvec$ in~\eqref{eq:sample_MB_quadratic}. Substituting $\etavec$ into~\eqref{eq:sample_MB_quadratic} gives two equivalent forms of the empirical \ac{mse},
\beqna
\label{eq:r1}
\widehat\Rcal(\hvec,\bvec)
\nonumber\\
=\tr(\widehat\Cxx)-2\hvec^T\sum_{m}b_m\mathbf r_m+\hvec^T\Big(\sum_{m,\ell}b_mb_\ell\Kmat_{m\ell}\Big)\hvec
\label{eq:r1_h}\\
=\tr(\widehat\Cxx)-2\bvec^T\mathbf q(\hvec)+\bvec^T\Gmat(\hvec)\bvec,
\label{eq:r1_b}
\eeqna
where $[\Gmat(\hvec)]_{m\ell}=\hvec^T\Kmat_{m\ell}\hvec$ and $[\mathbf q(\hvec)]_m=\hvec^T\mathbf r_m$. 
For fixed $\bvec$, \eqref{eq:r1_h} is minimized by the shared response~\eqref{eq:shared_solution} with $\omegavec=\bvec$. For fixed $\hvec$, \eqref{eq:r1_b} is minimized by $\bvec=\Gmat(\hvec)^{\dagger}\mathbf q(\hvec)$. Algorithm~\ref{alg:r1} alternates these two steps. Each step does not increase $\widehat\Rcal$, and the normalization $\mathbf 1^T\bvec=1$ does not change it, since $(\hvec/c)(c\bvec)^T=\hvec\bvec^T$. Thus, $\widehat\Rcal$ is nonincreasing over the iterations. Since the first step with $\bvec=\mathbf p$ yields the shared-response estimator, and the rank-one class is contained in the multi-basis class, the training \ac{mse} of the rank-one estimator lies between those of the multi-basis and shared-response estimators. Unlike in~\eqref{eq:shared_estimator}, the weights $\bvec$ are not constrained to be nonnegative. The estimator has $N+M-1$ free coefficients.

\begin{algorithm}[t]
\caption{Rank-one estimator (alternating least squares)}
\label{alg:r1}
\begin{algorithmic}[1]
\Require $\Kmat_{m\ell}$, $\mathbf r_m$ for $m,\ell=1,\ldots,M$; prior $\mathbf p$
\State $\bvec\gets\mathbf p$
\Repeat
\State $\hvec\gets\big(\sum_{m,\ell}b_mb_\ell\Kmat_{m\ell}\big)^{\dagger}\sum_m b_m\mathbf r_m$ \Comment{\textcolor{gray}{\eqref{eq:shared_solution}, $\omegavec=\bvec$}}
\State $\bvec\gets\Gmat(\hvec)^{\dagger}\mathbf q(\hvec)$, \ $\bvec\gets\bvec/(\mathbf 1^T\bvec)$
\Until{$\|\Delta\bvec\|\le10^{-8}\|\bvec\|$ or $50$ iterations}
\State Recompute $\hvec$ as in line 3
\Ensure $\Amat=\sum_m b_m\Vmat_m\diag(\hvec)\Vmat_m^T$
\end{algorithmic}
\end{algorithm}

\emph{Topology-selective and multiple-model GSP} (Subsection~\ref{subsec:topology_selection}). The branches $\hat\xvec_{\rm GSP}^{(m)}(\yvec)$ are~\eqref{eq:topology_specific_gsp} with~\eqref{eq:sample_topology_specific_response}. The likelihood in~\eqref{eq:map_topology_selector} is replaced by the graph-stationary Gaussian surrogate $\yvec|m^\star{=}m\sim\Ncal\big(\spmy{m},\Vmat_m\diag(\spdyy{m})\Vmat_m^T\big)$, which is diagonal in the GFT domain. Since $\Vmat_m$ is orthogonal, its log-determinant is $\sum_n\log[\spdyy{m}]_n$, and its quadratic form is a weighted sum of the squared GFT coefficients of $\yvec-\spmy{m}$. Hence, it is evaluated in $O(N^2)$ operations per candidate, without matrix inversion. For $P$ snapshots $\yvec^{(1)},\ldots,\yvec^{(P)}$ that share the topology ($P=1$ except in Subsection~\ref{ssec:selection}), the approximate posterior probabilities are
\begin{align}
\log\pi_m={}&\log p_m-\frac{P}{2}\sum_{n=1}^{N}\log[\spdyy{m}]_n\nonumber\\
&-\frac12\sum_{p=1}^{P}\sum_{n=1}^{N}\frac{\big[\Vmat_m^T(\yvec^{(p)}-\spmy{m})\big]_n^2}{[\spdyy{m}]_n}+c,
\label{eq:post_gsp}
\end{align}
where $c$ ensures $\sum_m\pi_m=1$. The topology-selective (hard) estimator applies $\hat\xvec_{\rm GSP}^{(\hat m)}$ with $\hat m=\arg\max_m\pi_m$, i.e.,~\eqref{eq:map_topology_selector} with the surrogate likelihood. The \emph{multiple-model} (soft) estimator is
\begin{equation}
\hat\xvec_{\rm MM}(\yvec)=\sum_{m=1}^{M}\pi_m\,\hat\xvec_{\rm GSP}^{(m)}(\yvec),
\label{eq:soft}
\end{equation}
a plug-in approximation of the topology-blind MMSE estimator~\eqref{eq:mmse_unknown}. Both estimators use $2MN$ coefficients ($\widehat\hvec_{{\rm GSP},m}$ and $\spdyy{m}$).

\section{Competing Topology-Blind Estimators}
\label{sec:competing}

\emph{Topology-blind LMMSE}: \eqref{eq:global_lmmse} with the empirical mixture moments and $\widehat{\Cyy}^{\dagger}$ in place of $\Cyy^{-1}$ ($N^2$ coefficients). In Subsection~\ref{ssec:training}, its regularized version applies~\eqref{eq:ridge} to the pooled, prior-weighted training data.

\emph{Nominal-topology GSP-LMMSE}~\cite{S:kroizer2021bayesian}: the branch $\hat\xvec_{\rm GSP}^{(1)}$ of the nominal topology, designed from $\mathcal D_1$ only and applied to all the data. This represents a design that ignores the topology uncertainty, and differs from the single-basis estimator of Section~\ref{sec:proposed}, which uses the basis $\Vmat_1$ but the mixture moments.

\emph{Expected-graph filter}~\cite{S:Isufi2017random}: the mixture-optimal graph filter in the eigenbasis $\bar\Vmat=[\bar{\mathbf v}_1,\ldots,\bar{\mathbf v}_N]$ of $\bar\Lmat=\sum_m p_m\Lmat_m$,
\begin{equation}
\hat\xvec=\widehat{\mubx}+\bar\Vmat\diag(\bar\hvec)\bar\Vmat^T(\yvec-\widehat{\muby}),\quad
[\bar\hvec]_n=\frac{\bar{\mathbf v}_n^T\widehat{\Cxy}\bar{\mathbf v}_n}{\bar{\mathbf v}_n^T\widehat{\Cyy}\bar{\mathbf v}_n},
\label{eq:expected}
\end{equation}
which uses $N$ coefficients.

\emph{Minimax MSE}~\cite{S:Eldar2005robust} and \emph{minimax regret}~\cite{S:Eldar2004competitive}, with the candidate family as the uncertainty set:
\begin{equation}
(\Amat,\bvec)=\argmin_{\Amat,\bvec}\ \max_{m=1,\ldots,M}\big(\spMSE_m(\Amat,\bvec)-c_m\big).
\label{eq:minimax}
\end{equation}
Here, $c_m=0$ for the minimax MSE estimator, and for the minimax regret estimator, $c_m$ is~\eqref{eq:mse_m} evaluated at the topology-aware LMMSE estimator of Section~\ref{sec:aware}, so that $\spMSE_m-c_m$ is the excess \ac{mse} over the oracle of topology $m$.

\emph{Saddle point.} The maximum of $M$ numbers equals the maximum of their convex combinations, so~\eqref{eq:minimax} is
\begin{equation}
\begin{split}
&\min_{\Amat,\bvec}\ \max_{\nubf\in\Delta_M}\sum_{m=1}^{M}\nu_m\big(\spMSE_m(\Amat,\bvec)-c_m\big).
\end{split}
\label{eq:F}
\end{equation}
Since the objective function in \eqref{eq:F} is a convex quadratic in $(\Amat,\bvec)$, linear in $\nubf$, and $\Delta_M$ is compact and convex, Sion's minimax theorem allows the min and max to be exchanged. For fixed $\nubf$, $\sum_m\nu_m\spMSE_m$ is the empirical counterpart of~\eqref{eq:bayes_risk} with the prior $\mathbf p$ replaced by $\nubf$, and the constants $c_m$ do not affect the minimizer. Hence, the inner minimizer is the topology-blind LMMSE estimator~\eqref{eq:global_lmmse} of the $\nubf$-mixture,
\begin{equation}
\Amat(\nubf)=\widehat{\Cxy}(\nubf)\widehat{\Cyy}(\nubf)^{\dagger},\quad
\bvec(\nubf)=\widehat{\mubx}(\nubf)-\Amat(\nubf)\widehat{\muby}(\nubf).
\label{eq:saddle}
\end{equation}
The outer problem maximizes $g(\nubf)\triangleq 
\sum_{m=1}^{M}\nu_m\big(\spMSE_m(\Amat(\nubf),\bvec(\nubf)-c_m\big)$ over $\Delta_M$. The function $g$ is concave, as a pointwise minimum of functions that are linear in $\nubf$. When $\widehat{\Cyy}(\nubf)\succ\mathbf 0$, the minimizer~\eqref{eq:saddle} is unique, and Danskin's theorem gives the gradient
\begin{equation}
\partial g/\partial\nu_m=\spMSE_m\big(\Amat(\nubf),\bvec(\nubf)\big)-c_m.
\label{eq:danskin}
\end{equation}
The minimax estimator is~\eqref{eq:saddle} at the least-favorable prior $\nubf^\star=\argmax_{\nubf\in\Delta_M}g(\nubf)$. By the optimality conditions on the simplex, all the topologies with $\nu_m^\star>0$ attain the maximal value of $\spMSE_m-c_m$. Thus, the estimator equalizes the worst-case loss over the hardest topologies: the least observable ones for the minimax MSE criterion, and those farthest from their oracle for the minimax regret criterion.

\emph{Algorithm.} We maximize $g$ by mirror ascent with the entropy potential, i.e., the exponentiated-gradient update with the gradient~\eqref{eq:danskin}, which keeps $\nubf$ on the simplex (Algorithm~\ref{alg:minimax}). Both estimators use $N^2$ coefficients.

\begin{algorithm}[t]
\caption{Minimax MSE / regret estimator (mirror ascent)}
\label{alg:minimax}
\begin{algorithmic}[1]
\Require Sample moments of all the topologies; $c_m$, $m=1,\ldots,M$
\State $\nubf\gets\mathbf 1/M$
\Repeat
\State Compute $\Amat,\bvec$ by~\eqref{eq:saddle}
\State $e_m\gets\spMSE_m(\Amat,\bvec)-c_m$, $m=1,\ldots,M$ \\
\Comment{\textcolor{gray}{gradient~\eqref{eq:danskin}}}
\State $\nubf\gets\nubf\odot\exp(\tau\,\mathbf e)$, \ $\nubf\gets\nubf/(\mathbf 1^T\nubf)$ \Comment{\textcolor{gray}{step size $\tau>0$}}
\Until{$g(\nubf)$ converges}
\State Compute $\Amat,\bvec$ by~\eqref{eq:saddle}
\Ensure $\Amat,\bvec$ \Comment{\textcolor{gray}{$\hat\xvec=\Amat\yvec+\bvec$}}
\end{algorithmic}
\end{algorithm}

\emph{Multiple-model, full LMMSE}~\cite{S:Magill1965,S:LiJilkov2005}: the classical multiple-model estimator. It is identical to~\eqref{eq:post_gsp}--\eqref{eq:soft}, but with unrestricted branches and full Gaussian likelihoods,
\begin{equation}
\begin{split}
\hat\xvec^{(m)}(\yvec)&=\spmx{m}+\Amat_m(\yvec-\spmy{m}),\\
\yvec|m^\star{=}m&\sim\Ncal\big(\spmy{m},\spCyy{m}+\gamma_m\Imat\big),
\end{split}
\label{eq:mm_full}
\end{equation}
where $\Amat_m$ and $\gamma_m$ are obtained from~\eqref{eq:ridge} with $\mathcal D_m$. In the likelihood, $\gamma_m$ is bounded below by $10^{-10}\tr(\spCyy{m})/N$ for numerical stability. The hard and soft rules are as in Section~\ref{sec:proposed}. The estimator uses about $\tfrac32MN^2$ coefficients ($\Amat_m$ and the symmetric $\spCyy{m}$).

\begin{table*}[t]
\centering
\caption{Compared estimators: definition, topology information at test time, structure, and number of learned coefficients.}
\label{tab:summary}
\renewcommand{\arraystretch}{1.1}
\begin{tabular}{@{}lllll@{}}
\toprule
Estimator & Definition & Topology at test time & Structure & Coefficients\\
\midrule
GSP-LMMSE, known topology~\cite{S:kroizer2021bayesian} & \eqref{eq:topology_specific_gsp}, $m=m^\star$ & revealed & graph filter of $\Lmat_{m^\star}$ & $MN$\\
LMMSE, known topology & \eqref{eq:conditional_lmmse}, \eqref{eq:ridge} & revealed & unrestricted, per topology & $MN^2$\\
\midrule
Multi-basis (Theorem~\ref{thm:MB}) & \eqref{eq:opt_est}, \eqref{eq:sample_eta} & none & $M$ graph filters & $MN$\\
Shared response (Proposition~\ref{prop:shared_blind}) & \eqref{eq:shared_estimator}, $\omegavec=\mathbf p$ & none & common response & $N$\\
Rank-one (Remark~\ref{rem:lowrank}) & \eqref{eq:lowrank_factor}, Alg.~\ref{alg:r1} & none & $\Hmat=\hvec\bvec^T$ & $N+M-1$\\
Single basis $\Vmat_{m_0}$ & \eqref{eq:shared_estimator}, $\omegavec=\mathbf e_{m_0}$ & none & graph filter of $\Lmat_{m_0}$ & $N$\\
Topology-selective / multiple-model GSP & \eqref{eq:post_gsp}, \eqref{eq:soft} & inferred & GSP branches, surrogate likelihood & $2MN$\\
\midrule
Topology-blind LMMSE & \eqref{eq:global_lmmse} & none & unrestricted & $N^2$\\
Nominal-topology GSP-LMMSE~\cite{S:kroizer2021bayesian} & \eqref{eq:topology_specific_gsp}, $m=1$ & assumes $m=1$ & graph filter of $\Lmat_1$ & $N$\\
Expected graph~\cite{S:Isufi2017random} & \eqref{eq:expected} & none & graph filter of $\bar\Lmat$ & $N$\\
Minimax MSE~\cite{S:Eldar2005robust} / regret~\cite{S:Eldar2004competitive} & \eqref{eq:minimax}, \eqref{eq:saddle}, Alg.~\ref{alg:minimax} & none & unrestricted & $N^2$\\
Multiple-model, full~\cite{S:Magill1965,S:LiJilkov2005} & \eqref{eq:mm_full} & inferred & full branches and likelihoods & $\approx\frac32MN^2$\\
DNN~\cite{S:Mestav2019} & \eqref{eq:dnn} & none & MLP $N$--$H$--$H$--$N$ & $2HN+H^2+2H+N$\\
\bottomrule
\end{tabular}
\end{table*}

\emph{Deep neural network}~\cite{S:Mestav2019}: a fully connected network that approximates $\mathrm E[\xvec|\yvec]$, trained on the pooled data of all the candidates without topology labels,
\beqna
\hat\xvec=\mathbf m_x+
\Dmat_x\hspace{5.75cm}\nonumber\\\times \Big(\Wmat_3\rho\big(\Wmat_2\rho(\Wmat_1\Dmat_y^{-1}(\yvec-\mathbf m_y)+\bvec_1)
+\bvec_2\big)+\bvec_3\Big)\hspace{-0.05cm},
\label{eq:dnn}
\eeqna
where $\rho(\cdot)=\max(\cdot,0)$, the two hidden layers have $H=256$ units, and $\mathbf m_x,\mathbf m_y$ and the diagonal $\Dmat_x,\Dmat_y$ are the entrywise means and standard deviations of the training data. The loss is the empirical \ac{mse} of the standardized states plus the weight decay $10^{-5}\sum_i\|\Wmat_i\|_F^2$. Training uses Adam~\cite{S:Kingma2015} (learning rate $10^{-3}$, mini-batches of $128$), He initialization~\cite{S:He2015}, and early stopping on a $15\%$ validation split (patience $15$, at most $200$ epochs). Since $T_m$ is equal for all $m$ and the prior is uniform, the pooled loss is the empirical counterpart of~\eqref{eq:bayes_risk}. The network has $2HN+H^2+2H+N$ parameters.

\vspace{-0.25cm}
\section{Cross-Network Designs}
\label{sec:transfer}

In Subsection~\ref{ssec:transfer}, each group $j$ contains one known graph (Remark~\ref{rem:spectral_interpretation}), and the coefficients are the minimizer of the weighted curve fit~\eqref{eq:curvefit_risk} with the empirical responses and powers.
For each network, we use $2\times10^4$ training and $5\times10^3$ test
samples, and the results are averaged over five independent data sets.
As in Remark~\ref{rem:spectral_interpretation}, the DC term $n=1$ ($\lambda_{j,1}=0$) is excluded, and the penalty $\sum_q\kappa_q\alpha_q^2$ is added to~\eqref{eq:curvefit_risk}. The learned function is applied to the held-out network by~\eqref{eq:unseen_estimator}. The compared dictionaries are:
\begin{itemize}
\item \emph{Chebyshev}: $\phi_q(\lambda)=T_q(\lambda/s_j-1)$, $Q=60$, and $\kappa_q=10^{-7}q^4$, on three axes: absolute, $s_j=\max_i\lambda_{\max,i}/2$ for all $j$; normalized as in ChebNet~\cite{S:Defferrard2016}, $s_j=\lambda_{\max,j}/2$; and graphon-normalized~\cite{S:Ruiz_Chamon_Ribeiro_GSP_2021}, $s_j=N_j\max_i(\lambda_{\max,i}/N_i)/2$, i.e., an axis proportional to $\lambda/N_j$.
\item \emph{Rational}: $\phi_q(\lambda)=1/(\lambda+c_q)$, with four poles $c_q$ log-spaced in $[10^{-2},10^{2}]$ and $\kappa_q=10^{-8}$.
\item \emph{Linearized response} (no learning): $h(\lambda)=\beta/(\beta\lambda+\sigma^2)$, the Wiener response of the linearized model M1.
\end{itemize}
In the few-shot experiment, $T_0$ labeled samples of the held-out network are drawn at random from its training set ($50$ draws for each $T_0$ and data set), and four designs are compared: the GSP-LMMSE estimator~\eqref{eq:sample_topology_specific_response} fitted on the $T_0$ samples ($N$ coefficients); the rational function fitted on the $T_0$ samples; the rational function transferred from the other five networks, with the target means estimated from the $T_0$ samples; and the rational function fitted jointly on the six groups by~\eqref{eq:sample_crosssize_alpha}, where the target group contributes its $T_0$ samples. The excess \ac{mse} is measured relative to the GSP-LMMSE estimator designed on the full training set of the held-out network, and averaged over the six networks.

\vspace{-0.25cm}
\section{Summary}
Table~\ref{tab:summary} summarizes the compared estimators. The means are additionally estimated and are not counted as coefficients.


\end{document}

%% file: myshorts.tex
\newcommand{\avec}{{\bf{a}}}

\newcommand{\bvec}{{\bf{b}}}

\newcommand{\yvec}{{\bf{y}}}

\newcommand{\wvec}{{\bf{w}}}
\newcommand{\xvec}{{\bf{x}}}

\newcommand{\svec}{{\bf{s}}}

\newcommand{\gvec}{{\bf{g}}}

\newcommand{\hvec}{{\bf{h}}}

\newcommand{\etavec}{{\bf{\eta}}}

\newcommand{\onevec}{{\bf{1}}}
\newcommand{\zerovec}{{\bf{0}}}

\newcommand{\alphavec}{{\bf{\alpha}}}

\newcommand{\Lambdamat}{{\bf{\Lambda}}}

\newcommand{\Amat}{{\bf{A}}}

\newcommand{\Bmat}{{\bf{B}}}
\newcommand{\Cmat}{{\bf{C}}}
\newcommand{\Dmat}{{\bf{D}}}

\newcommand{\Gmat}{{\bf{G}}}
\newcommand{\Hmat}{{\bf{H}}}

\newcommand{\Imat}{{\bf{I}}}
\newcommand{\Lmat}{{\bf{L}}}
\newcommand{\Kmat}{{\bf{K}}}

\newcommand{\Umat}{{\bf{U}}}
\newcommand{\Vmat}{{\bf{V}}}
\newcommand{\Wmat}{{\bf{W}}}

\newcommand{\Ymat}{{\bf{Y}}}

\newcommand{\define}{\stackrel{\triangle}{=}}





\def\alphavec{{\mbox{\boldmath $\alpha$}}}

\def\etavec{{\mbox{\boldmath $\eta$}}}

\def\alphavecsmall{{\mbox{\boldmath {\scriptsize $\alpha$}}}}

\DeclareMathOperator*{\argmin}{argmin}
\DeclareMathOperator*{\argmax}{argmax}

\newcommand{\diag}{\mathrm{diag}}

\newcommand{\be}{\begin{equation}}
\newcommand{\ee}{\end{equation}}
\newcommand{\beqna}{\begin{eqnarray}}
\newcommand{\eeqna}{\end{eqnarray}}

%% file: Preamble.tex
\definecolor{NewColor}{rgb}{0,0,0}

\acrodef{pmf}[pmf]{probability mass function}
\acrodef{pdf}[pdf]{probability density  function}
\acrodef{psd}[PSD]{positive-semidefinite}
\acrodef{evd}[EVD]{eigenvalue decomposition}
\acrodef{gft}[GFT]{graph Fourier transform}
\acrodef{gsp}[GSP]{Graph signal processing}
\acrodef{fir}[FIR]{finite impulse response}
\acrodef{ot}[OT]{optimal-transport}

\acrodef{dnn}[DNN]{deep neural network}
\acrodef{map}[MAP]{maximum a-posteriori probability}
\acrodef{snr}[SNR]{signal-to-noise ratio}
\acrodef{mse}[MSE]{mean-squared-error}
\acrodef{amse}[AMSE]{averaged MSE}
\acrodef{mmse}[MMSE]{minimum MSE}

\acrodef{gsplmmse}[GSP-LMMSE]{GSP-linear minimum mean-squared-error}
\acrodef{gnn}[GNN]{graph neural network}
\acrodef{arma}[ARMA]{auto-regressive moving average}
\acrodef{ista}[ISTA]{iterative soft thresholding algorithm}
\acrodef{lista}[LISTA]{learned \ac{ista}}
\acrodef{lmmse}[LMMSE]{linear MMSE}
\acrodef{dsp}[DSP]{digital signal processing}
\acrodef{gso}[GSO]{graph shift operator}
\acrodef{wrt}[w.r.t.]{with respect to}
\acrodef{KKT}[KKT]{Karush-Kuhn-Tucker}
\acrodef{psse}[PSSE]{power system state estimation}
\acrodef{gsp}[GSP]{Graph signal processing}
\acrodef{gft}[GFT]{graph Fourier transform}
\acrodef{ghpf}[GHPF]{Graph High Pass Filter}
\acrodef{gso}[GSO]{graph shift operator}
\acrodef{sdgs}[SDGS] {sparse diffused graph signal}
\acrodef{wrt}[w.r.t.]{with respect to}
\acrodef{arma}[ARMA]{autoregressive moving average}
\acrodef{iid}[i.i.d.]{independent and identically distributed}

%% file: tables/tab_A_networks.tex
\begin{table}[t]
\centering
\caption{Test systems: buses, transmission lines, extreme nonzero graph frequencies, and the root-mean-square (RMS) phase (state) angle, $\sqrt{\mathrm{E}[\|\mathbf{x}\|^2]/N}$, under models M1 and M2.}
\label{tab:networks}
\footnotesize
\setlength{\tabcolsep}{3pt}
\begin{tabular}{@{}lrrrrrr@{}}
\hline
& & & & & \multicolumn{2}{c}{RMS angle [$^\circ$]}\\
System & $N$ & Lines & $\lambda_2$ & $\lambda_{\max}$ & M1 & M2\\
\hline
IEEE 14 & 14 & 20 & 2.07 & 60 & 6.4 & 5.2\\
IEEE RTS 24 & 24 & 38 & 3.22 & 243 & 3.9 & 10.4\\
IEEE 30 & 30 & 41 & 0.67 & 105 & 8.0 & 4.5\\
New England (NE) 39 & 39 & 46 & 5.12 & 1,025 & 2.8 & 5.5\\
IEEE 57 & 57 & 80 & 0.35 & 119 & 8.8 & 4.6\\
IEEE 118 & 118 & 186 & 0.30 & 580 & 6.1 & 7.8\\
\hline
\end{tabular}
\end{table}

%% file: tables/tab_A_methods.tex
\begin{table}[t]
\centering
\caption{Compared estimators, their number of learned
coefficients, and design principles. All estimators
use the same labeled training data (Section~V).}
\label{tab:methods}
\fontsize{7.75}{9}\selectfont
\setlength{\tabcolsep}{1pt}
\begin{tabular}{@{}llcl@{}}
\hline
Estimator & Ref. & Coeff. & Principle\\
\hline
\multicolumn{4}{l}{\emph{Topology aware (benchmarks)}}\\
GSP-LMMSE (known top.)
& \cite{kroizer2021bayesian} & $MN$
& active-topology filter\\
LMMSE (known top.)
& (14) & $MN^2$
& unrestricted, per topology\\

\multicolumn{4}{l}{\emph{Proposed (topology blind)}}\\
Multi-basis
& Thm.~\ref{thm:MB} & $MN$
& filters in all $M$ bases\\
Rank-one
& Rem.~\ref{rem:lowrank} & $N{+}M{-}1$
& $\Hmat=\hvec\bvec^T$, learned $\bvec$\\
Shared response
& Prop.~\ref{prop:shared_blind} & $N$
& common response, $\boldsymbol\omega=\mathbf p$\\
Selective GSP
& \eqref{eq:map_topology_selector} & $2MN$
& approximate-MAP selection\\
Multiple-model GSP
& \eqref{eq:40} & $2MN$
& posterior-weighted branches\\

\multicolumn{4}{l}{\emph{Competing (topology blind)}}\\
Blind LMMSE
& \eqref{eq:global_lmmse} & $N^2$
& unrestricted operator\\
Nominal-topology GSP
& \cite{kroizer2021bayesian} & $N$
& nominal-topology filter\\
Expected graph
& \cite{Isufi2017random} & $N$
& filter of $\bar{\mathbf L}$\\
Minimax MSE
& \cite{Eldar2005robust} & $N^2$
& minimizes $\max_m\mathrm{MSE}_m$\\
Minimax regret
& \cite{Eldar2004competitive} & $N^2$
& minimizes maximum regret\\
Multiple-model, full
& \cite{Magill1965,LiJilkov2005}
& $\frac{3}{2}MN^2$
& full covariances (hard/soft)\\
DNN
& \cite{Mestav2019} & $2HN{+}H^2$
& MLP, $N$--$H$--$H$--$N$\\
\hline
\end{tabular}
\end{table}

%% file: tables/tab_B_comparison.tex
\begin{table}[t]
\centering
\caption{Estimator comparison under topology uncertainty (M1, $20$\,dB SNR, $M=20$, uniform prior): mean and worst-topology NMSE [dB] over $20$ realizations (95\% confidence half-widths of the mean ${\le}\,0.40$\,dB), and design time on the IEEE 118-bus system. Proposed estimators are in bold.}
\label{tab:comparison}
\footnotesize
\setlength{\tabcolsep}{3pt}
\begin{tabular}{@{}lccccc@{}}
\hline
 & \multicolumn{2}{c}{IEEE 118-bus} & \multicolumn{2}{c}{NE 39-bus} & Time\\
Estimator & Mean & Worst & Mean & Worst & [s]\\
\hline
\multicolumn{6}{@{}l}{\emph{Topology aware (benchmarks)}}\\
GSP-LMMSE (known top.) & $-9.64$ & $-8.36$ & $-19.66$ & $-17.67$ & $0.0039$\\
LMMSE (known top.) & $-8.47$ & $-7.15$ & $-19.32$ & $-17.32$ & $0.68$\\
\multicolumn{6}{@{}l}{\emph{Topology blind, linear}}\\
Blind LMMSE & $-8.53$ & $-6.34$ & $-9.26$ & $-4.93$ & $0.0045$\\
\textbf{Multi-basis} & $-8.56$ & $-6.29$ & $-9.26$ & $-4.92$ & $0.53$\\
\textbf{Shared response} & $-8.45$ & $-5.90$ & $-8.89$ & $-3.65$ & $0.082$\\
\textbf{Rank-one} & $-8.53$ & $-6.07$ & $-9.20$ & $-4.67$ & $0.72$\\
Nominal-topology GSP & $-8.45$ & $-5.52$ & $-8.04$ & $-1.74$ & $0.00035$\\
Expected graph & $-8.58$ & $-6.37$ & $-9.27$ & $-4.94$ & $0.0044$\\
Minimax MSE & $-7.62$ & $-6.45$ & $-8.14$ & $-6.12$ & $1.2$\\
Minimax regret & $-7.94$ & $-6.47$ & $-8.18$ & $-6.26$ & $1.9$\\
\multicolumn{6}{@{}l}{\emph{Topology blind, nonlinear (data dependent)}}\\
\textbf{Selective GSP} & $-8.30$ & $-6.77$ & $-8.82$ & $-6.94$ & $0.0043$\\
\textbf{Multiple-model GSP} & $-8.83$ & $-6.97$ & $-11.31$ & $-7.76$ & $0.0041$\\
Multiple-model, full (hard) & $-7.18$ & $-5.69$ & $-8.76$ & $-6.63$ & $0.67$\\
Multiple-model, full & $-7.67$ & $-6.01$ & $-10.69$ & $-7.33$ & $0.67$\\
DNN & $-5.73$ & $-4.65$ & $-8.21$ & $-5.27$ & $66$\\
\hline
\end{tabular}
\end{table}

%% file: tables/tab_F_loo.tex
\begin{table*}[t]
\centering

\caption{Leave-one-network-out spectral transfer under M1
(SNR $=20$\,dB) and M2 (SNR $=40$\,dB).
The ``Own'' columns give the topology-aware GSP-LMMSE
NMSE; all other columns give the excess MSE relative
to this benchmark (dB).
The Wiener baseline uses the fixed response
$\beta/(\beta\lambda+\sigma^2)$.
Results are averaged over five realizations.}
\label{tab:loo}
\vspace{-0.15cm}
\footnotesize
\setlength{\tabcolsep}{3pt}
\begin{tabular}{l|rrrrrr|rrrrr}
\hline
& \multicolumn{6}{c}{M1, SNR $=20$\,dB}
& \multicolumn{5}{c}{M2, SNR $=40$\,dB}\\
Target network
& Own & Cheb. & Cheb.\ $\lambda/\lambda_{\max}$
& Cheb.\ $\lambda/N$ & Rational & Wiener
& Own & Cheb. & Cheb.\ $\lambda/\lambda_{\max}$
& Cheb.\ $\lambda/N$ & Rational\\
& NMSE & abs. & \cite{Defferrard2016}
& \cite{Ruiz_Chamon_Ribeiro_GSP_2021}
& & (lin.)
& NMSE & abs. & \cite{Defferrard2016}
& \cite{Ruiz_Chamon_Ribeiro_GSP_2021} & \\
\hline
IEEE 14
& -14.73 & $0.46$ & $13.50$ & $12.37$ & $0.00$ & $0.01$
& -16.62 & $2.53$ & $15.73$ & $12.08$ & $1.45$\\
IEEE RTS 24
& -18.05 & $0.06$ & $14.92$ & $13.16$ & $0.00$ & $0.00$
& -28.75 & $10.86$ & $25.78$ & $23.13$ & $7.32$\\
IEEE 30
& -10.92 & $0.26$ & $8.87$ & $6.57$ & $0.01$ & $0.02$
& -11.37 & $0.87$ & $9.46$ & $4.95$ & $2.41$\\
New England 39
& -20.00 & $12.39$ & $28.85$ & $32.94$ & $0.01$ & $0.00$
& -29.97 & $5.62$ & $37.23$ & $17.88$ & $3.63$\\
IEEE 57
& -8.26 & $0.58$ & $6.33$ & $2.05$ & $0.07$ & $0.14$
& -8.19 & $3.56$ & $6.47$ & $4.97$ & $1.48$\\
IEEE 118
& -9.81 & $1.39$ & $6.64$ & $12.86$ & $0.12$ & $0.01$
& -15.93 & $12.60$ & $13.93$ & $16.05$ & $5.89$\\
\hline
\end{tabular}
\end{table*}